\pdfoutput=1
\documentclass[sigconf]{acmart}
\usepackage{graphicx}
\usepackage{textcomp}
\usepackage[dvipsnames]{xcolor}
\usepackage{amsmath,amsthm,amsfonts}
\usepackage{thm-restate}
\usepackage{bbm}
\usepackage{xspace}
\usepackage{tikz}
\usepackage{color}
\usepackage{hyperref}
\usepackage{physics}
\usepackage[inline]{enumitem}
\usepackage{balance}

\usepackage[noend]{algpseudocode}
\usepackage{algorithm}
\usepackage{subcaption}
\usepackage{booktabs}
\usepackage{multirow}
\usepackage{pdfpages}
\usepackage{cleveref}
\newcommand{\card}[1]{\left|#1\right|}

\newcommand{\rk}[1]{\mathfrak {#1}}

\usepackage{hyperref}

\newtheorem{claim}{Claim}
\newtheorem{remark}{Remark}

\usepackage{tikz}
\usetikzlibrary{positioning, fit, calc}
\usetikzlibrary{svg.path}
\usetikzlibrary{shapes.geometric}
\usetikzlibrary{patterns}
\usetikzlibrary{arrows.meta} %

\setenumerate[1]{label={(\arabic*)}}

\definecolor{blue}{RGB}{0,50,200}
\definecolor{magenta}{RGB}{255,0,255}
\hypersetup{colorlinks,linkcolor=blue,urlcolor=blue,citecolor=magenta}

\usetikzlibrary{positioning}

\newcommand{\prm}{\textsc{Circles}}
\newcommand{\tiereport}{\textsc{Circles\_TieReport}}
\newcommand{\tieshare}{\textsc{Circles\_TieSharing}}
\newcommand{\tiebreak}{\textsc{Circles\_TieBreak}}
\newcommand{\prmo}{\textsc{Ord-Circles}}
\newcommand{\po}{\textsc{Ordering}}

\newcommand{\N}{\mathbb{N}}
\newcommand{\ceil}[1]{\lceil#1\rceil}

\usepackage{etoolbox}

\newcommand{\appsymb}{$\star$}
\newcommand{\appref}[1]{{\hyperref[proof:#1]{\appsymb}}}

\newcommand{\appendixproof}[2]{%
  \gappto{\appendixProofs}
  {
    \subsection{Proof of \Cref{#1}}\label{proof:#1}
    #2
  }
}

\begin{document}

\title{Ranking Opinions with Few States in Population Protocols}

\author{Tom-Lukas Breitkopf}
\orcid{0009-0008-2875-1945}
\affiliation{%
  \institution{%
  Technische Universität Berlin}
  \city{Berlin}
  \country{Germany}
}
\email{t.breitkopf@tu-berlin.de}

\author{Julien Dallot}
\orcid{0009-0008-1286-1373}
\affiliation{%
  \institution{%
  Technische Universität Berlin}
  \city{Berlin}
  \country{Germany}}
\email{julien.dallot@tu-berlin.de}

\author{Antoine El-Hayek}
\orcid{0000-0003-4268-7368}
\affiliation{%
  \institution{Institue of Science and Technology Austria}
  \city{Klosterneuburg}
  \country{Austria}
}
\email{antoine.el-hayek@ist.ac.at}

\author{Stefan Schmid}
\orcid{0000-0002-7798-1711}
\affiliation{%
  \institution{%
  Technische Universität Berlin}
  \city{Berlin}
  \country{Germany}}
\email{stefan.schmid@tu-berlin.de}

\renewcommand{\shortauthors}{Breitkopf, Dallot, El-Hayek and Schmid}

\acmYear{2026}\copyrightyear{2026}
\setcopyright{cc}
\setcctype[4.0]{by}
\acmConference[PODC '26]{ACM Symposium on Principles of Distributed Computing}{July 6--10, 2026}{Egham, United Kingdom}
\acmBooktitle{ACM Symposium on Principles of Distributed Computing (PODC '26), July 6--10, 2026, Egham, United Kingdom}
\acmDOI{10.1145/3796701.3815913}
\acmISBN{979-8-4007-2512-8/26/07}

\begin{abstract}
  Population protocols are a model of distributed computing where $n$ agents, each a simple finite-state machine, interact in pairs to solve a common task against a (adversarial) interaction scheduler.
  This model was intensively studied in recent years; in particular, the problem of relative majority received much attention: Each agent starts with an input opinion (or color) out of~$k$ possibilities, and the goal is for each agent to eventually output the color with the largest support in the population.
  Before our work, the state complexity (the minimum number of states required per agent) was only known to be between $\Omega(k^2)$ and $O(k^{7})$.
  Our main contribution is a population protocol that solves the relative majority problem with $k^3$ states.
  We achieve this result with a new protocol called \textsc{Circles}.
  While prior approaches in the literature relied on duels of agents to find the majority color --- an approach that proved effective for the case with two colors~---~\textsc{Circles} partitions the agents into circular linked lists of decreasing sizes, with the property that no two agents with the same initial color lie in the same circle.
  We show that \textsc{Circles} always correctly computes the desired structure against the most adversarial of schedulers (weakly fair).
  We then show that a trivial extension of \textsc{Circles} solves the relative majority problem.
  We extend our protocol to handle various tie-breaking mechanisms or to support the case where the agents do not share a prior ordering of the colors.
  Finally, we show that a modification of \textsc{Circles} solves the ranking problem with~$2 \cdot k^4$ states, where each agent must output the rank of its initial color in the population.
\end{abstract}



\maketitle

\section{Introduction}
\label{sec:intro}

Population protocols are a distributed computing model where a population of $n$ finite-state agents interacts in a chaotic manner, updating their state according to a common protocol upon interacting with each other.
Ever since their introduction in 2006 by Angluin et al.~\cite{angluin2006networksoffinitestatesensors}, population protocols have been widely studied and have become a standard setting to model key aspects of distributed computing~\cite{alistarh2024gamedynamicsinpopulationprotocols}.
Initially used as an abstraction for passively mobile finite-state sensor networks \cite{angluin2006networksoffinitestatesensors}, population protocols features intriguing connections to natural processes \cite{cardelli2012cellcycleswitchcomputesapproximatemajority}, and can be used to model certain chemical reactions~\cite{doty2014timingincrns}, dynamics in social groups or animal populations, or even gene regulatory networks~\cite{BBH01}.
The model further features connections to other theoretical models such as Petri nets, vector addition systems, semilinear sets and Presburger arithmetic \cite{angluin2006networksoffinitestatesensors}.
Multiple problems have already been studied in this model, the most illustrative of which are the leader election and majority problem~\cite{elsasser2018recentresultsinpp}.

This paper focuses on two related problems in the population protocol model, \textit{relative majority} and \textit{ranking}.
Both assume that initially every agent stores a color.
Relative majority then asks for each agent to output the most popular color, i.e., the one with the largest initial support in the population.
The ranking problem generalizes relative majority by asking each agent to output the position of its initial color ranked by popularity.
Both problems are fundamental tasks in distributed computing \cite{becchetti2015pluralityconsensus} in order to ensure fault tolerance or to enable conflict resolution.
The majority problem is also relevant outside computer science, with applications in physics, biology, statistics and sociology~\cite{berenbrink2016efficientpluralityconsensus}.
The ranking problem amounts to implementing a sorting algorithm in the population protocol model, which is a basic task in many subfields of computer science.

We are particularly interested in the fundamental question of how many states a protocol requires to solve a given task: the state complexity.
This question is of relevance as in practice, the number of states that can be used is severely limited~\cite{alistarh2018recentalgorithmicadvancesinpp}: tiny sensor nodes have a limited memory~\cite{angluin2006networksoffinitestatesensors}, chemical applications ideally require a small number of chemical species involved in the reactions~\cite{czerner2023lowerboundsonthestatecomplexityofpp}, and in DNA implementations the risk of faulty behavior increases with the number of states~\cite{alistarh2018recentalgorithmicadvancesinpp}.

\subsection{Related work}
\label{ssec:related_work}

The majority problem is one of the most studied problems in the population protocol model, along with leader election~\cite{elsasser2018recentresultsinpp}.
Majority is a special case of relative majority with two colors: the agents can have one of two colors as input and the goal is for each agent to output the color with the larger support in the population.
There exists a simple protocol solving majority with only $4$ states per agent, concurrently found in \cite{doi:10.1137/110823018} and \cite{mertzios2014determiningmajoritynetworkslocal}; this protocol is provably the best possible in terms of state complexity~\cite{10.1145/2767386.2767429}.
A lot of research focuses on the trade-off between time and state complexity~\cite{10.1145/2767386.2767429, alistarh2017timespacetradeoffspopulationprotocols, elsaesser21, alistarh2017spaceoptimalmajoritypopulationprotocols}.
This culminated in a protocol with $O(\log n)$ many states, solving majority in expected $O(n\log n)$ interactions assuming in every step of the execution two agents interact uniformly at random. This protocol by \citet{DBLP:conf/focs/DotyEGSUS21} is optimal space-wise for protocols running in $O(n\log n)$ time.

A first study of multiple colors in a related model was conducted in 2015 by~\citet{salehkaleybar2015distributedvoting}.
They give solutions for both relative majority and ranking.
Their definition of the ranking problem differs from ours, as they require every agent to output the whole ranking, not just the position of its own color in the ranking.
Their algorithms require an exponential number of states: $O(k \cdot 2^{k})$ states per agent for the relative majority problem and $O(k \cdot k!)$ states for the ranking problem, where $k$ is the number of possible colors.
They show that this bound is tight for the ranking problem -- intuitively, an agent of any color ($k$ many) should be able to output any possible ranking ($k!$ many).
We note that the exponential state complexity of this protocol is not only implied by its output requirement, but also by its internal logic. %
The authors conjectured that also the bound on the state complexity of relative majority if tight.
This conjecture was disproved in 2017 by G\k{a}sieniec et al.~\cite{gasieniec2017deterministicppforexactmajorityandpluarlity} who first broke the exponential barrier and gave a protocol that solves the relative majority problem with $O(k^7)$ states (or as frequently cited~$O(k^6)$ states in case the agents only have a binary output indicating whether their initial color is in relative majority).
The protocol orchestrates a distributed tournament between colors: the initial colors of each agent are written in binary form and define a binary tournament tree according to the longest common prefix rule.
The color that wins the tournament is deemed to be the majority color.

The above protocol relies on a hard-coded binary representation of the colors, which implies that the agents share a prior ordering on the colors.
Motivated by biological use cases, Natale and Ramezani~\cite{natale2019necessarymemorycomputeplurality} proposed a modified model where no such prior ordering is known, allowing agents only to memorize colors and to compare them only for equality (a formal definition of a related model where agents are only able to store a single immutable color can be found in~\cite{blondin2023populationprotocolsunordereddata}).
They proposed a protocol using $O(k^{11})$ states in the unordered model by combining the relative majority protocol of G\k{a}sieniec et al. with another protocol to assign a numerical label to each color.
\citet{natale2019necessarymemorycomputeplurality} also show that any protocol solving relative majority requires~$\Omega(k^2)$ states.

Apart from the work listed above concerned with always-correct protocols, there is also research on protocols correctly converging with high probability.
This relaxation allows providing protocols with~$O(k)$ to $O(k+\log n)$ states, achieving different guarantees of validity depending on the initial bias between the majority opinion and the second most popular opinion~\cite{DBLP:conf/soda/BankhamerBBEHKK22,DBLP:conf/podc/BankhamerBBEHKK22,DBLP:conf/podc/AmirABBHKL23,DBLP:conf/podc/El-HayekE025}. The additive bias required is typically $\Omega(\sqrt {n\log n})$ which motivates us in the search for an always-correct protocol.

\newcommand{\tspace}{\rule{0pt}{2.6ex}}
\newcommand{\bspace}{\rule[-0.9ex]{0pt}{0pt}}
\begin{table}[t]
  \centering
  \label{tab:related-work}
  \begin{tabular}{|l|l|c|}
    \hline
    \textbf{Model} & \textbf{Reference} & \textbf{State comp.} \tspace\bspace\\
    \hline
    \multirow{4}{*}{Ordered}
      & \citeauthor{salehkaleybar2015distributedvoting}~\cite{salehkaleybar2015distributedvoting} (2015) & $O(k \cdot 2^k)$ \tspace\\
    \cline{2-3}
      & \citeauthor{gasieniec2017deterministicppforexactmajorityandpluarlity}~\cite{gasieniec2017deterministicppforexactmajorityandpluarlity} (2017) & $O(k^7)$ \tspace\\
    \cline{2-3}
      & \textbf{This work} & $\mathbf{k^3}$ \tspace\\
    \cline{2-3}
      & Best-known lower bound~\cite{natale2019necessarymemorycomputeplurality} & $\Omega(k^2)$ \tspace\\
    \hline
    \multirow{2}{*}{Unordered}
      & \citeauthor{natale2019necessarymemorycomputeplurality}~\cite{natale2019necessarymemorycomputeplurality} (2019) & $O(k^{11})$ \tspace\\
    \cline{2-3}
      & \textbf{This work} & $\mathbf{O(k^4)}$ \tspace\\
    \hline
  \end{tabular}
  \vspace{0.2cm}
  \caption{
    Existing results on the state complexity to solve the relative majority problem with probability $1$ and for $k$ colors.
  }
  \vspace{-0.85cm}
\end{table}

\subsection{Our contribution}

This work presents a new always-correct protocol, \textsc{Circles}, that can be adapted to solve both the relative majority and the ranking problems in the population protocol model.
Unlike previous approaches for this problem~\cite{salehkaleybar2015distributedvoting, gasieniec2017deterministicppforexactmajorityandpluarlity, natale2019necessarymemorycomputeplurality}, our protocol does not rely on pairwise duels between colors but rather on a novel way to structure the agents into circular linked lists, ensuring that the smallest lists contains colors with the greatest support.
This approach enables us to improve upon the state of the art regarding the state complexity for both problems we consider.
For relative majority, we achieve a state complexity of~$O(k^3)$ under a weakly-fair scheduler~(see~\Cref{def:weakly_fair_scheduler}):

\begin{restatable}[The \textsc{Circles} protocol and relative majority]{theorem}{circrelat}
  The \prm{} protocol, which uses $k^2$ many states, denoted $\braket{i}{j}$ for $i,j \in [k]$, stabilizes under a weakly-fair scheduler.
  In the stable configuration, agents of the form $\braket{i}{i}$ exist if and only if $i$ is the majority color.
  Propagating the result to every agent can be achieved
by adding one output variable to the protocol, for a total of $k^3$ states.
\end{restatable}

For the ranking problem, we break the exponential barrier~\cite{salehkaleybar2015distributedvoting} and provide a protocol that uses $O(k^4)$ under a globally fair scheduler~(see \Cref{def:globally_fair_scheduler}):

\begin{restatable}{theorem}{thmranking}\label{thm:ranking}
    There is a protocol that solves the ranking problem in case of a globally fair scheduler, and uses $2\cdot k^4$ many states.
\end{restatable}

Previous results combined several sub-protocols in an intricate way.
In contrast, the \prm{} protocol is surprisingly simple, making it easy to analyse and adapt to different use cases.
In this work we not only use it to solve both the ranking problem and the relative majority problem, but also a variety of modifications of the relative majority problem, dropping the assumption that initially an ordering on the colors exists and handling ties in a number of different ways.

\newcommand{\ang}[1]{\left\langle#1\right\rangle}

\subsection{Organization}

The remainder of this paper is structured as follows.
In Section~\ref{sec:model_problem} we define the population protocol model, the relative majority problem in different versions and the ranking problem.
In Section~\ref{sec:circle} we present our new protocol \textsc{Circles} adapted to solve relative majority and prove its correctness.
In Section \ref{sec:tie} we present how to adapt \textsc{Circles} to solve the relative majority problem while handling ties in different ways.
In Section~\ref{sec:unordered} we show how to adapt \textsc{Circles} to the unordered data model, where the agents do not share a prior ordering of the colors.
In Section~\ref{sec:ranking} we show how to adapt \textsc{Circles} to solve the ranking problem, and in \Cref{sec:future_work}, we discuss the gap between the $O(k^3)$ upper bound and $\Omega(k^2)$ lower bound for the problem.

\section{Model, Problems and Notations}
\label{sec:model_problem}
We will now briefly introduce models and problems used within this work.

\subsection{Models and Problems}

\paragraph{The Population Protocol model}
\label{ssec:model_spp}
Unless stated otherwise, we use the standard notation from~\cite{angluin2006networksoffinitestatesensors}.
A \textit{population} $\mathcal{P}$ consists of a set $A$ of $n$ agents and an irreflexive relation $E \subseteq A \times A$.
It can be seen as a directed interaction graph.
In this work we only consider the complete interaction graph.
When a \textit{population protocol} $\mathcal{A}$ runs in a population $\mathcal{P}$ every agent $a \in A$ is in one state of a finite set of states.
The input function maps an agents input, taken from some finite input set, to an initial state.
In this work we always consider an input set of $k$ colors.
After the initialization, pairs of agents interact.
The agents have distinct roles when they interact, one is the \textit{initiator} and the other is the \textit{responder}.
An interaction causes the initiator and the responder to update their states according to the transition function, which depends only on their previous states. 
Finally, the output of the population protocol is given by an \textit{output function} which maps at any time an agent's state to an output, taken from a finite output set.

A \textit{configuration} captures the ``global state'' of a population by assigning a state to each agent.
Formally, a configuration can be fully described as a multiset of states.
Similarly an input (or output) assignment maps an input (output) to every agent of the population.
An interaction between two agents~$a, b$ in configuration~$S$ results in a new configuration~$S'$, written as:~$S \xrightarrow{(a,b)} S'$.
Configuration $S''$ is \textit{reachable} from $S$ if there is some sequence of interactions, such that a population starting in $S$ undergoing the interaction sequence results in  $S''$.
We write this as~$S\rightarrow S''$.\\

A \textit{computation} of a protocol is an infinite fair sequence of interactions.
Which two agents interact in each step is determined by a \textit{scheduler}.
We introduce two different schedulers formalizing two different notions of fairness as done in other work~\cite{natale2019necessarymemorycomputeplurality,ABBS16}.
\begin{definition}[Weakly Fair Scheduler]
  \label{def:weakly_fair_scheduler}
  In a computation under a \textit{weakly fair scheduler} every pair of agents interacts infinitely often.
\end{definition}
\begin{definition}[Globally Fair Scheduler]
  \label{def:globally_fair_scheduler}
  In a computation under a \textit{globally fair scheduler} for any two configurations $S$ and $S'$ with~$S\rightarrow S'$ it holds that if $S$ occurs infinitely often in the computation, then $S'$ also occurs infinitely often.
\end{definition}

A computation \textit{stabilizes} to some output assignment when no future interaction can change the output of any agent.

\paragraph{The Unordered Data model}
In the unordered data model, we assume that the agents do not agree on a prior ordering of the input colors: The agents can only compare their inputs for equality and memorize colors.
This setting was first used in~\cite{natale2019necessarymemorycomputeplurality} and a related version later formalized in~\cite{blondin2023populationprotocolsunordereddata}.

\paragraph{Problems}
In the literature the relative majority problem is also referred to as plurality, majority consensus or proportionate agreement~\cite{becchetti2015pluralityconsensus}.
The input assignment assigns a color from~$\{0,1,\dots,k-1\}$ to every agent.
A protocol solving the problem must stabilize to an output assignment in which every agent indicates as output the color that has more supporters than any other.
Different versions on how to handle ties will be introduced later.

In the ranking problem, the input assignment gives each agent a color from~$\{0,1,\dots,k-1\}$ as well.
A protocol solving the problem must stabilize to an output assignment in which each agent indicates as output the position of its input color in a ranking of the colors: an agent with the plurality color as its initial color should output $1$, an agent with the second most supported color should output $2$, and so on.

\subsection{Notations}

We explain here the notations used throughout the paper.
Proofs marked with~\appsymb{} are deferred to the appendix.

\paragraph{multisets. }
  For sets $S$ and $T$ we write $S^{T}$ to denote the set of functions $f : T \to S$.
  If $T$ is finite we call the elements of $\mathbb{N}^T$ \emph{multisets} over $T$.
  In this paper, the subset $\subseteq$, the union $\cup$ and the set subtraction $\setminus$ operations are generalized to multisets: if $A, B \in \mathbb{N}^T$ are two multisets on $T$, it holds that $A \subseteq B$ if for all $t \in T$, $A(t) \leq B(t)$, we define $A \cup B$ as the multiset $C$ such that for all $t \in T$, $C(t) = A(t) + B(t)$, we define $A \setminus B$ as the multiset $C$ such that for all $t \in T$, $C(t) = \max(0, A(t) - B(t))$.

\paragraph{remainder. }
  For $x \in \mathbb{Z}$ and $p \in \mathbb{N}^*$, we define $x \bmod p$ as the remainder of the euclidean division of $x$ by $p$. Note that this is a number in $\mathbb N$, not $\mathbb Z/p\mathbb Z$.

\paragraph{ranges. }
  Let $x, y \in \mathbb{N}$ such that $x \le y$.
  $[x, y]$ denotes the set $\{x, x+1, \dots, y-1, y\}$, and $(x, y)$ denotes the set $\{x+1, x+2, \dots, y-2, y-1\}$.
  With~$[y]$ we denote~$\{0,1,\dots,y\}$.

\paragraph{modulo ranges.}
\sloppy
  Let $x, y \in \mathbb{N}$.
  We define modulo ranges:  $[x, y]_p$ denotes the set $\{x \bmod p, (x+1) \bmod p, \dots, (x+(y-x) \bmod p) - 1) \bmod p, (x + (y-x) \bmod p) \bmod p\}$, and $(x, y)_p$ denotes the set $\{(x+1) \bmod p , (x+2) \bmod p, \dots, (x + (y-x) \bmod p) - 2) \bmod p, (x + (y - x) \bmod p) - 1) \bmod p\}$.
  For instance, $[2,7]_{10} $ $= \{2, 3, 4, 5, 6, 7\}$ and $(8, 3)_{10} = \{9, 0, 1, 2\}$.

  We say that a modulo range $[x,y]_p$ is included in another one $[x', y']_p$, denoted $[x,y]_p \subseteq [x',y']_p$ if it is included in the set-theoretic sense, with the extra condition: $x \neq y'$ or $y \neq x'$. This is to avoid the corner case: $[a,a+1]_p \subseteq [a+1,a]_p$, see~\Cref{fig:circle-graduations} for an illustration.

\begin{figure}[h]
  \centering
  \begin{tikzpicture}[scale=.6]
    \def\radius{1.3}
    \draw (0,0) circle (\radius);
    \foreach \i in {0,...,9} {
      \pgfmathsetmacro{\angle}{90 - \i * 36}
      \pgfmathsetmacro{\x}{\radius * 0.75 * cos(\angle)}
      \pgfmathsetmacro{\y}{\radius * 0.75 * sin(\angle)}
      \node at (\x, \y) {\i};
      \pgfmathsetmacro{\xout}{\radius * cos(\angle)}
      \pgfmathsetmacro{\yout}{\radius * sin(\angle)}
      \pgfmathsetmacro{\xin}{(\radius - 0.15) * cos(\angle)}
      \pgfmathsetmacro{\yin}{(\radius - 0.15) * sin(\angle)}
      \draw (\xin, \yin) -- (\xout, \yout);
    }
    \draw[thick, black] (-162:\radius +0.2) arc (-162:-306:\radius +0.2);
    \draw[thick, black] (-162:\radius +0.1) -- (-162:\radius +0.3);
    \draw[thick, black] (-306:\radius +0.1) -- (-306:\radius +0.3);
  \end{tikzpicture}
  \caption{
    A modulo range contains all numbers covered by a clockwise walk from the start to the end of the range around the modular clock.
    This figure depicts $[7, 1]_{10}$.
  }
  \label{fig:circle-graduations}
\end{figure}
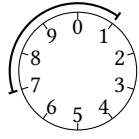

\paragraph{braket. }
   We overload the braket notation, frequently used in quantum mechanics, to note ordered pairs. For $i, j \in \N$, we write $\braket{i}{j}$ simply to distinguish the different roles of $i$ and $j$. For agents storing some pair $\braket{i}{j}$, we say they store a braket and refer to their bra and their ket for $i$ and $j$ respectively.

\section{Solving the Relative Majority Problem}
\label{sec:circle}

\subsection{Presentation of the protocol}

In this section we introduce our protocol that solves the relative majority problem and prove its correctness.
We consider the set $[0, k-1]$ of $k$ many colors and present below the \prm{}\footnote{We overload the notation slightly, and also refer by \prm{} the protocol that behaves similarly but does not have an \emph{out} entry. This is useful to use the protocol for other purposes than relative majority.} protocol.

  \label{prot:circle}
  \textbf{States:} The set of states $Q$ contains every triples $(i, j, o) \in [0, k-1]^{3}$.
  In the remainder of this paper we will use the braket notation $\braket{i}{j}$ to refer to the first two numbers of the triple, \emph{bra} refers to $\bra i$ and \emph{ket} refers to $\ket j$, while \emph{out} refers to $o$. 

  \textbf{Input:} each agent is initialized with $\braket{i}{i}$ and $\emph{out} = i$, where $i \in [0, k-1]$ is the input color of the agent.

  \textbf{Output:} return $\emph{out}$

  \textbf{Transition function:} We define \emph{weights} for each braket $\braket{i}{j}$ as follows:
  \begin{equation*}
  	w(\braket{i}{j}) =
  	\begin{cases}
  		k & \textit{if } i = j\\
  		(j - i) \bmod k& \textit{otherwise}
  	\end{cases}
  \end{equation*}
  Two agents $a$ and $b$ that interact perform two successive operations:
\begin{enumerate}
\item
  $a$ and $b$ exchange their kets in case doing so strictly decreases the minimum weight of their two brakets.
\item
  If exactly one of $a$ or $b$ is of the form $\braket{i}$ for some $i \in [0, k-1]$, set $\emph{out}_a=\emph{out}_b=i$.
\end{enumerate}

\medskip
The intuition of the protocol is depicted in Figure~\ref{fig:intuition}.
Let's assume in this paragraph, for the sake of intuition, that only one color is in relative majority.
As shown in Figure~\ref{fig:intuition}, think of a population as stacks of colors with the height of each stack representing the number of agents supporting it.
If one were to remove rows from the stacks bottom to top until only a single row is left, then that row would contain
only the color in relative majority.
This is roughly what our protocol is doing, though it does not actually remove rows. Instead, in a way, it organizes the agents of a row into a cyclic linked list (which we call circle) and 
does not let them escape from it anymore.
The \emph{bra} of an agent represents its original color, while the \emph{ket} of the agent represents the next color in the circle.
The minimizing operation in the transition function ensures that if an agent is linked to another one too far away in the circle, it can always correct itself by meeting the correct \emph{ket}. 
It organizes agents into a circle that is as tightly packed as possible (with respect to the numerical representation of their color). Once no agents can be added anymore, the circle never changes.
The procedure then continues for the remaining agents, until in the final circle only one agent is contained: one holding the majority opinion.

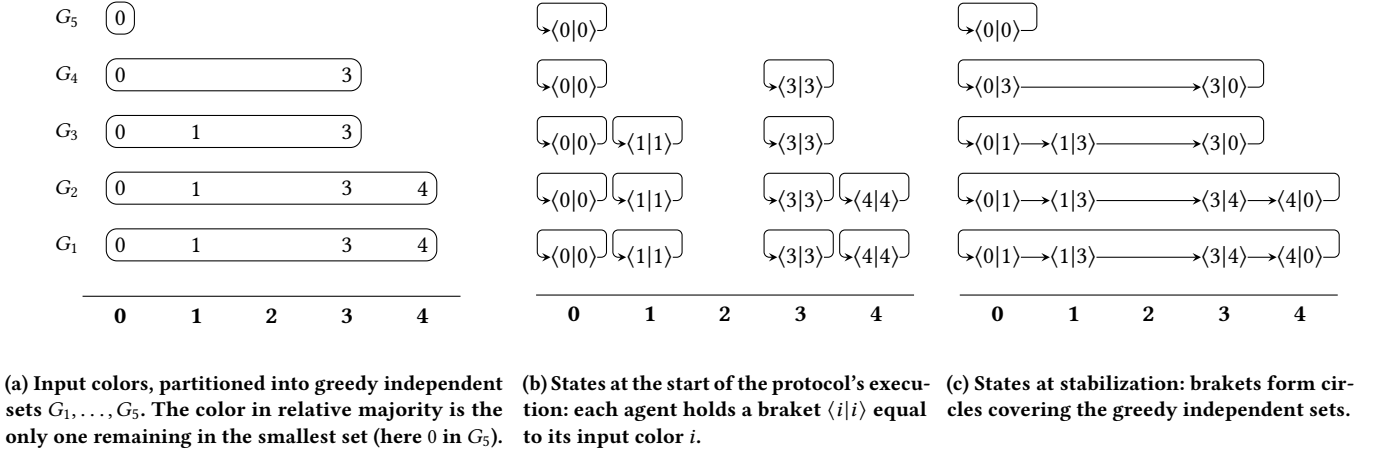
\begin{figure*}[h]
  \begin{subfigure}[t]{0.37\textwidth}
	\begin{minipage}[t][5cm][c]{\linewidth}
	\centering
	\begin{tikzpicture}[plain/.style={}]
		\node [plain]  at (0, 0.25){\textbf{0}};
		\node [plain] at (1, 0.25){\textbf{1}};
		\node [plain] at (2, 0.25){\textbf{2}};
		\node [plain] at (3, 0.25){\textbf{3}};
		\node [plain] at (4, 0.25){\textbf{4}};
		\draw[black] (-.5, .5) -- (4.5, .5);
		
		\pgfmathsetmacro{\yoff}{0.18}
		\node [plain] (o1) at (0, 1 + \yoff) {0};
		\node [plain] (o2) at (0, 1.75 + \yoff) {0};
		\node [plain] (o3) at (0, 2.5 + \yoff) {0};
		\node [plain] (o4) at (0, 3.25 + \yoff) {0};
		\node [plain] (o5) at (0, 4 + \yoff) {0};
		
		\node [plain] (t1) at (1, 1 + \yoff) {1};
		\node [plain] (t2) at (1, 1.75 + \yoff) {1};
		\node [plain] (t3) at (1, 2.5 + \yoff) {1};
		
		\node [plain] (fo1) at (3, 1 + \yoff) {3};
		\node [plain] (fo2) at (3, 1.75 + \yoff) {3};
		\node [plain] (fo3) at (3, 2.5 + \yoff) {3};
		\node [plain] (fo4) at (3, 3.25 + \yoff) {3};

		\node [plain] (fi1) at (4, 1 + \yoff) {4};
		\node [plain] (fi2) at (4, 1.75 + \yoff) {4};
		
		\node [rectangle, draw=black, rounded corners, inner sep=0pt, fit = (o1) (fi1)] {};
		\node [plain] at (-0.7, 1 + \yoff) {\textcolor{black}{\small{$G_1$}}};
		
		\node [rectangle, draw=black, rounded corners, inner sep=0pt, fit = (o2) (fi2)] {};
		\node [plain] at (-0.7, 1.75 + \yoff) {\textcolor{black}{\small{$G_2$}}};
		
		\node [rectangle, draw=black, rounded corners, inner sep=0pt, fit = (o3) (fo3)] {};
		\node [plain] at (-0.7, 2.5 + \yoff) {\textcolor{black}{\small{$G_3$}}};
		
		\node [rectangle, draw=black, rounded corners, inner sep=0pt, fit = (o4) (fo4)] {};
		\node [plain] at (-0.7, 3.25 + \yoff) {\textcolor{black}{\small{$G_4$}}};
		
		\node [rectangle, draw=black, rounded corners, inner sep=0pt, fit = (o5)] {};
		\node [plain] at (-0.7, 4 + \yoff) {\textcolor{black}{\small{$G_5$}}};
	\end{tikzpicture}
	\end{minipage}
	\caption{
          Input colors, partitioned into greedy independent sets $G_1,\dots,G_5$.
          The color in relative majority is the only one remaining in the smallest set (here $0$ in $G_5$).}
  \end{subfigure}\hfill
  \begin{subfigure}[t]{0.30\textwidth}
	\begin{minipage}[t][5cm][c]{\linewidth}
	\centering
	\begin{tikzpicture}[plain/.style={}]
		\node [plain]  at (0, 0.25){\textbf{0}};
		\node [plain] at (1, 0.25){\textbf{1}};
		\node [plain] at (2, 0.25){\textbf{2}};
		\node [plain] at (3, 0.25){\textbf{3}};
		\node [plain] at (4, 0.25){\textbf{4}};
		\draw[black] (-.5, .5) -- (4.5, .5);
		
		\node [plain] (o1) at (0, 1) {$\braket{0}$};
		\node [plain] (o2) at (0, 1.75) {$\braket{0}$};
		\node [plain] (o3) at (0, 2.5) {$\braket{0}$};
		\node [plain] (o4) at (0, 3.25) {$\braket{0}$};
		\node [plain] (o5) at (0, 4) {$\braket{0}$};
		
		\node [plain] (t1) at (1, 1) {$\braket{1}$};
		\node [plain] (t2) at (1, 1.75) {$\braket{1}$};
		\node [plain] (t3) at (1, 2.5) {$\braket{1}$};
		
		\node [plain] (fo1) at (3, 1) {$\braket{3}$};
		\node [plain] (fo2) at (3, 1.75) {$\braket{3}$};
		\node [plain] (fo3) at (3, 2.5) {$\braket{3}$};
		\node [plain] (fo4) at (3, 3.25) {$\braket{3}$};

		\node [plain] (fi1) at (4, 1) {$\braket{4}$};
		\node [plain] (fi2) at (4, 1.75) {$\braket{4}$};

		\foreach \n in {o1,o2,o3,o4,o5,t1,t2,t3,fo1,fo2,fo3,fo4,fi1,fi2}
			\draw[->, -stealth, rounded corners=2pt]
				($(\n.east)+(-3.5pt,0)$)
				-- ($(\n.east)+(-0pt,0)$)
				-- ($(\n.north east)+(-0pt,3pt)$)
				-- ($(\n.north west)+(-1.5pt,3pt)$)
				-- ($(\n.west)+(-1.5pt,0)$)
				-- ($(\n.west)+(3.5pt,0)$);
		
	\end{tikzpicture}
	\end{minipage}
	\caption{States at the start of the protocol's execution: each agent holds a braket $\braket{i}$ equal to its input color~$i$.}
  \end{subfigure}\hfill
  \begin{subfigure}[t]{0.30\textwidth}
	\begin{minipage}[t][5cm][c]{\linewidth}
	\centering
	\begin{tikzpicture}[plain/.style={}]
		\node [plain]  at (0, 0.25){\textbf{0}};
		\node [plain] at (1, 0.25){\textbf{1}};
		\node [plain] at (2, 0.25){\textbf{2}};
		\node [plain] at (3, 0.25){\textbf{3}};
		\node [plain] at (4, 0.25){\textbf{4}};
		\draw[black] (-.5, .5) -- (4.5, .5);
		
		\node [plain] (o1) at (0, 1) {$\braket{0}{1}$};
		\node [plain] (o2) at (0, 1.75) {$\braket{0}{1}$};
		\node [plain] (o3) at (0, 2.5) {$\braket{0}{1}$};
		\node [plain] (o4) at (0, 3.25) {$\braket{0}{3}$};
		\node [plain] (o5) at (0, 4) {$\braket{0}$};
		
		\node [plain] (t1) at (1, 1) {$\braket{1}{3}$};
		\node [plain] (t2) at (1, 1.75) {$\braket{1}{3}$};
		\node [plain] (t3) at (1, 2.5) {$\braket{1}{3}$};
		
		\node [plain] (fo1) at (3, 1) {$\braket{3}{4}$};
		\node [plain] (fo2) at (3, 1.75) {$\braket{3}{4}$};
		\node [plain] (fo3) at (3, 2.5) {$\braket{3}{0}$};
		\node [plain] (fo4) at (3, 3.25) {$\braket{3}{0}$};

		\node [plain] (fi1) at (4, 1) {$\braket{4}{0}$};
		\node [plain] (fi2) at (4, 1.75) {$\braket{4}{0}$};

		\draw[->,-stealth, rounded corners=2pt]
			($(fi1.east)+(-3.5pt,0)$) -- ($(fi1.east)+(2.5pt,0)$)
			-- ($(fi1.north east)+(2.5pt,3pt)$) -- ($(o1.north west)+(-2.5pt,3pt)$)
			-- ($(o1.west)+(-2.5pt,0)$) -- ($(o1.west)+(3.5pt,0)$);
		\draw[->,-stealth, rounded corners=2pt]
			($(fi2.east)+(-3.5pt,0)$) -- ($(fi2.east)+(2.5pt,0)$)
			-- ($(fi2.north east)+(2.5pt,3pt)$) -- ($(o2.north west)+(-2.5pt,3pt)$)
			-- ($(o2.west)+(-2.5pt,0)$) -- ($(o2.west)+(3.5pt,0)$);
		\draw[->,-stealth, rounded corners=2pt]
			($(fo3.east)+(-3.5pt,0)$) -- ($(fo3.east)+(2.5pt,0)$)
			-- ($(fo3.north east)+(2.5pt,3pt)$) -- ($(o3.north west)+(-2.5pt,3pt)$)
			-- ($(o3.west)+(-2.5pt,0)$) -- ($(o3.west)+(3.5pt,0)$);
		\draw[->,-stealth, rounded corners=2pt]
			($(fo4.east)+(-3.5pt,0)$) -- ($(fo4.east)+(2.5pt,0)$)
			-- ($(fo4.north east)+(2.5pt,3pt)$) -- ($(o4.north west)+(-2.5pt,3pt)$)
			-- ($(o4.west)+(-2.5pt,0)$) -- ($(o4.west)+(3.5pt,0)$);
		\draw[->,-stealth, rounded corners=2pt]
			($(o5.east)+(-3.5pt,0)$) -- ($(o5.east)+(2.5pt,0)$)
			-- ($(o5.north east)+(2.5pt,3pt)$) -- ($(o5.north west)+(-2.5pt,3pt)$)
			-- ($(o5.west)+(-2.5pt,0)$) -- ($(o5.west)+(3.5pt,0)$);

		\draw[->, -stealth] ($(o1.east)+(-3.5pt,0)$) -- ($(t1.west)+(3.5pt,0)$);
		\draw[->, -stealth] ($(t1.east)+(-3.5pt,0)$) -- ($(fo1.west)+(3.5pt,0)$);
		\draw[->, -stealth] ($(fo1.east)+(-3.5pt,0)$) -- ($(fi1.west)+(3.5pt,0)$);
		\draw[->, -stealth] ($(o2.east)+(-3.5pt,0)$) -- ($(t2.west)+(3.5pt,0)$);
		\draw[->, -stealth] ($(t2.east)+(-3.5pt,0)$) -- ($(fo2.west)+(3.5pt,0)$);
		\draw[->, -stealth] ($(fo2.east)+(-3.5pt,0)$) -- ($(fi2.west)+(3.5pt,0)$);
		\draw[->, -stealth] ($(o3.east)+(-3.5pt,0)$) -- ($(t3.west)+(3.5pt,0)$);
		\draw[->, -stealth] ($(t3.east)+(-3.5pt,0)$) -- ($(fo3.west)+(3.5pt,0)$);
		\draw[->, -stealth] ($(o4.east)+(-3.5pt,0)$) -- ($(fo4.west)+(3.5pt,0)$);

	\end{tikzpicture}
	\end{minipage}
	\caption{
          States at stabilization: brakets form circles covering the greedy independent sets.
        }
  \end{subfigure}
  \caption{
    Example of \prm{} with $k=5$ colors and $n=14$ agents.
    Agents with the same input color are stacked vertically.
  }
	\label{fig:intuition}
        \label{fig:circles}
        \Description[Bar chart of the input and linked lists that form circles]{The input is represented as a bar chart, where the most common input has the longest bar. Initially, there are as many linked lists as there are agents. At stabilisation, there are as many linked lists as there are agents with the most common input.}
\end{figure*}

\subsection{Proof of Correctness}
\label{sec:protocol_proof}

We prove the protocol's correctness in Theorems \ref{theorem:stabilization} and \ref{theorem:correctness}.
We beforehand introduce Greedy Independent sets, a construction on the input colors used to prove the protocol's correctness, as well as two necessary, fairly easy Lemmas \ref{lemma:majority_color} and \ref{lemma:global_braket_invariant}. This is a formalisation of rows in \Cref{fig:circles}.
  
\begin{definition}[Greedy Independent Sets]
  \label{def:greedy_independent_sets}
  Consider the multiset of input colors to our protocol.
  We partition this multiset into sets $G_1, G_2, \dots, G_q$ as follows: store in $G_1$ as many inputs as possible as long as $G_1$ does not contain two equal colors; then apply the same on the remaining inputs to obtain $G_2$, and so on.
\end{definition}

We begin with this fairly straightforward observation:

\begin{restatable}[\appref{lemma:majority_color}]{lemma}{majoritycolor}
  \label{lemma:majority_color}
  (Majority Color).
  Assume there exists a unique color $\mu$ in relative majority, then it holds that $G_q = \{\mu\}$ and there is no $j \neq \mu$ and $p \in [1, q]$ such that $G_p = \{j\}$.
\end{restatable}

\appendixproof{lemma:majority_color}{\begin{proof}
  Each time a set $G_p$ is constructed (following the construction process in Definition \ref{def:greedy_independent_sets}), any color whose count is not zero yet is added into $G_p$ and its count is decremented by one.
  The color $\mu$ appears in the population strictly more than any other, it therefore holds that $\forall p \in [1, q]$, $\mu \in G_p$.
  Let now $i \in [0, k-1]$ be a color contained in a set $G_p$ for some $p \in [1, q]$.
  It holds that $\forall l \le p$, $i \in G_l$ as because color $i$ is available to populated $G_p$, it was also available when $G_l$ was filled earlier.
  It is therefore not possible that a color $j \neq \mu$ is contained in $G_q$ as $j$ would be contained in all sets and thus $j$ would also be in relative majority, contradicting $\mu$ being the unique color in relative majority.
\end{proof}}

We then introduce the main invariant of \prm{}, which stems from the fact that agents only exchange kets.

\begin{restatable}[\appref{lemma:global_braket_invariant}]{lemma}{invariant}
  \label{lemma:global_braket_invariant}
  (Global Braket Invariant).
  In every configuration and for all $i \in [0, k-1]$, the number of bras $\bra{i}$ and the number of kets $\ket{i}$ are equal.
\end{restatable}

\appendixproof{lemma:global_braket_invariant}{\begin{proof}
  Every agent is initialized with braket $\braket{i}{i}$ for some $i \in [0, k-1]$ and the claim initially holds.
  Agents subsequently only ever update their braket, by exchanging kets among each other.
  The overall number of bras and kets in the population therefore does not change during a computation.
\end{proof}}

\begin{restatable}{theorem}{stabilization}
  \label{theorem:stabilization}
  (Stabilization).
  Running \textsc{Circles}, the agents exchange their kets a finite number of times.
\end{restatable}
\begin{proof}
  We call $\omega$ the smallest ordinal number greater than all the integers\footnote{For readers unfamiliar with ordinals, we recommend checking \href{https://en.wikipedia.org/wiki/Ordinal_number}{this wikipedia page}, or \href{https://www.youtube.com/watch?v=uWwUpEY4c8o}{this 15-minute video}, which contains all that is needed for our proofs.}.
  We prove the claim by exhibiting a quantity that strictly decreases at each ket exchange.
  Given a configuration $C$, let $w_1(C), w_2(C) \dots w_n(C)$ be the weights of all the brakets sorted in increasing order.
  Define
  \begin{align*}
    g(C) =  \omega^{n-1} \cdot w_1(C) + \omega^{n-2} \cdot w_2(C)  + \dots + \omega \cdot w_{n-1}(C) +  w_{n}(C)
  \end{align*}
  Assume that two agents exchange their kets.
  Let $p$ be the lowest index such that $w_p(C)$ changes before and after the ket exchange.
  By design of the protocol, $w_p(C)$ strictly decreases.
  This implies that $g$ strictly decreases when two agents exchange their kets.
  As an ordinal number cannot decrease infinitely many times, the number of ket exchanges is therefore finite.
\end{proof}

\begin{restatable}{theorem}{correctness}
  \label{theorem:correctness}
  (Correctness).
  Assume that there exists a unique color $\mu$ in relative majority.
  In the \prm{} protocol, all agents eventually output $\mu$ under a weakly fair scheduler.
\end{restatable}

We define the brakets associated with a Circle:

\begin{definition}[Circle Braket Sets]\label{def:braketsets}
    For a given greedy independent set $G_p$ with $p \in [1,q]$ (Definition~\ref{def:greedy_independent_sets}), let $g_0, g_1, \dots ,g_m$ be the elements of $G_p$ sorted in increasing order and define $$f(G_p) = \{\braket{g_0}{g_1}, \braket{g_1}{g_2}, \dots , \braket{g_m}{g_0}\}$$
\end{definition}

Theorem~\ref{theorem:correctness} follows directly from the following lemma:

\begin{lemma}\label{lem:majowins}
    After Stabilization (Theorem~\ref{theorem:stabilization}), let $\mathcal C$ be the multiset of brakets of the agents. We have that:
    $$\mathcal{C} = \bigcup_{p=1 \dots q} f(G_p)$$
\end{lemma}

\begin{proof}

 We prove the following predicate $H(r)$ by induction on $r \in [0, q]$:
  \begin{align}
    \bigcup_{p = 1 \dots r} f(G_p) \subseteq \mathcal{C} \tag{$H(r)$}
  \end{align}

  \noindent
  The base case for $r=0$ is trivial.\\
  Let $r \in [0, q-1]$, we assume that $H(r)$ holds and show that $H(r+1)$ also holds.
  We define the subconfiguration $\mathcal{C}[r+1] = \mathcal{C} \setminus \cup_{p = 1 \dots r} f(G_p)$.

  Case 1: $\cup_{p=r+1}^q G_p$ contains only elements from one color.
  Let $i$ be that color.
    
    Then for any other color $j \neq i$, there are at most $r$ many bras $\bra{j}$ and as many kets $\ket{j}$, which are all included in $\cup_{p = 1 \dots r} f(G_p)$. Thus all agents in $\mathcal C[r+1]$ are of the form $\braket{i}$. Since
    $\{\braket{i}\} = f(G_{r+1})$, we have $ \cup_{p = 1 \dots r+1} f(G_p) \subseteq \mathcal{C}$.

  Case 2: There are at least two different colors in $G_{r+1}$.
  
  We note $g_0, g_1, \dots, g_m$ the elements of $G_{r+1}$ sorted in increasing order.
  Let $l \in [0, m]$. To lighten notations, we will mean $l + s \bmod (m+1)$ each time we write $l+s$ in the remainder of this proof.
  We prove that, if there is no agent with braket $\braket{g_l}{g_{l+1}}$ in $C[r+1]$, then there exist two agents whose interaction creates that braket, a contradiction with the stability hypothesis.
  First notice that $\bra{g_l}$ and $\ket{g_{l+1}}$ are in $\mathcal{C}[r+1]$.
  Indeed, note that there are at least $r+1$ many $\bra{g_l}$ and $r+1$ many $\bra{g_{l+1}}$ in $\mathcal C$, as there are at least $r+1$ many agents with color $g_l$ and $g_{l+1}$ initially.
  By Theorem~\ref{lemma:global_braket_invariant}, this means we have at least $r+1$ many $\ket{g_{l+1}}$ in $\mathcal C$.
  Because each $f_p$ for $p \le r$ contains exactly one $\bra{g_l}$ and one $\ket{g_{l+1}}$, we have that both $\bra{g_l}$ and $\ket{g_{l+1}}$ are in $\mathcal{C}[r+1]$.

  Assuming by contradiction that there is no agent with braket $\braket{g_l}{g_{l+1}}$ in $\mathcal{C}[r+1]$, then there exists an agent with braket $\braket{g_l}{j}$ and an agent with braket $\braket{i}{g_{l+1}}$ in $\mathcal{C}[r+1]$ for some $i$ and $j$.
  We show that those two agents exchange their kets if they interact.

  \begin{claim}\label{claim:notininterval}
      $i, j \notin (g_l, g_{l+1})_m$
  \end{claim}
  \begin{proof}
      By contradiction, assume that $i$ is in $(g_l, g_{l+1})_m$. 
      A $\bra{i}$ in $\mathcal C[r+1]$ indicates that $i$ had initially at least $r+1$ agents supporting it, as by contradiction if it was not the case, all the bras $\bra{i}$ would have been in $\cup_{p=1}^r f(G_p)$.
      By construction of $G_{r+1}$, we must have that $i$ is in $G_{r+1}$, and thus, $g_l$ and $g_{l+1}$ are not consecutive in the ordered list of $G_{r+1}$, a contradiction.
      The case $j \in (g_l, g_{l+1})_m$ is symmetric.
  \end{proof}
  If $i \neq g_{l+1}$, it holds by Claim~\ref
{claim:notininterval} that
  \begin{multline*}
    w(\braket{g_l}{g_{l+1}})= (g_{l+1} - g_l) \bmod k  <  (g_{l+1} - i) \bmod k \\ = w(\braket{i}{g_{l+1}})
  \end{multline*}
  Otherwise, if $i=g_{l+1}$:
  \begin{align*}
    w(\braket{g_l}{g_{l+1}})= (g_{l+1} - g_l) \bmod k  <  k  = w(\braket{i}{g_{l+1}})
  \end{align*}

  Similarly, if $j \neq g_l$, it holds by Claim~\ref
{claim:notininterval} that
  \begin{align*}
    w(\braket{g_l}{g_{l+1}}) = (g_{l+1} - g_l) \bmod k  <  (j - g_{l}) \bmod k   = w(\braket{g_l}{j})
  \end{align*}
  Otherwise, if $j=g_{l}$:
  \begin{align*}
    w(\braket{g_l}{g_{l+1}})= (g_{l+1} - g_l) \bmod k  <  k  = w(\braket{g_l}{j})
  \end{align*}
  
  Which proves that exchanging kets between $\braket{g_l}{j}$ and $\braket{i}{g_{l+1}}$ reduces the minimum weight.  
  An interaction between the two agents eventually happens as the scheduler is weakly fair, this interaction would therefore trigger a ket exchange, which is in contradiction with the stability hypothesis: we deduce that $\braket{g_l}{g_{l+1}} \in \mathcal{C}[r+1]$, which implies $f(G_{r+1}) \subseteq \mathcal{C}[r+1]$ and therefore $H(r+1)$ holds.\\
  We proved by induction that $H(q)$ holds:
  \begin{align*}
    \bigcup_{p = 1 \dots q} f(G_p) \subseteq \mathcal{C}
  \end{align*}
  As $|\cup_{p = 1 \dots q} f(G_p)| = |\mathcal{C}|$ we can rewrite $H(q)$ as an equality and the claim holds.
\end{proof}

\begin{proof}[Proof of Theorem~\ref{theorem:correctness}]
    By Lemma~\ref{lem:majowins} and Lemma~\ref{lemma:majority_color}, after Stabilization (Theorem~\ref{theorem:stabilization}), since we assumed that there is only one majority color, there exists at least one agent in braket $\braket{\mu}$ and none in braket $\braket{j}$ for $j \neq \mu$.
    The agent(s) with braket $\braket{\mu}{\mu}$ will transmit their output color $\mu$ to the rest of the population and the claim follows.
    
The claim on the number of states needed is straightforward, as $3$ numbers between $0$ and $k-1$ are stored.
In the case where agents only need to output whether or not their color is the majority one, an extra bit needs to be stored to record that fact, and the $\braket{i}$ agents need to flip it accordingly when meeting other agents.
\end{proof}

This yields the following theorem:

\circrelat*

Note that before this last proof, at no point did we use the fact that there is only one majority color, and thus, we can use all the lemmas to analyze what happens in the case of ties, which we do in the following sections.

\subsection{A note on the time complexity}
While this paper focuses on the space complexity of the problem, it is worthwhile talking briefly about the time complexity of the problem, which is typically done under the \emph{random scheduler}: where for each interaction, two agents are chosen uniformly at random.
Before talking about the \prm{} protocol, let's look at an intermediate result: Assume you have $\ell$ pairs of tokens on some agents that need to meet, in any order. Then in any interaction, the probability that one such pair is chosen is $\frac {2\ell} {n(n-1)}$, and thus the expected time until one such pair being chosen is $\frac {n(n-1)}{2\ell}$. Then there are only $\ell-1$ many pairs to interact. Thus, the total expected time for all pairs to interact is $\sum_{\lambda =1}^\ell \frac {n(n-1)}{2\lambda} =\Theta(n^2\log \ell)$.

Now consider the \prm{} protocol. Assume for the moment that all opinions initially have the same number of agents. Then there are $n$ bra-ket pairs that need to meet before stabilisation, in any order. By the result above, this takes $O(n^2\log n)$ time.

If we remove the assumption that all opinions initially have the same number of agents, then order matters. Indeed, the circles with more agents need to have stabilised before circles with fewer agents can start pairing up. There are at most $k$ many different circles, and thus, the total running time is at most $O(kn^2\log n)$.

\section{Tie Handling Mechanisms}
\label{sec:tie}
In this section, we describe three different ways to handle ties in the relative majority problem.

\subsection{Reporting Ties}
\label{subsec:tie_reporting}

We present the \tiereport{} protocol which extends the \prm{} protocol and enables it to report if there is more than one color in relative majority, i.e. if there is a tie.
\tiereport{} solves a version of the relative majority problem where, in case of a tie, all the agents must return a specific value to indicate that there is a tie.
This feature was present in the paper from G\k{a}sieniec et al.~\cite{gasieniec2017deterministicppforexactmajorityandpluarlity} where the agents could keep track of a potential tie with the label zero.
In case there is no tie, all the agents must still return the majority color.

We add to the agents' states a variable tie\_state on top of the states of \prm{}.
tie\_state can have three values: \emph{no\_tie}, \emph{tie} or \emph{tie\_spreader}.
At the beginning, all the agents are at \emph{no\_tie}.
Without interfering, tie\_state evolves in parallel of \prm{} and follows three rules.

\begin{enumerate}
\item
  \textbf{tie detection. }
  An agent $\braket{i}{i}$ that decays into $\braket{i}{j}$ with $i \neq j$ goes to \emph{tie\_spreader}.
\item
  \textbf{tie spreading. }
  An agent at \emph{tie\_spreader} sets to \emph{tie} any encountered agent $\braket{i}{j}$ such that $i \neq j$.
\item
  \textbf{tie correction. }
  An agent $\braket{i}{i}$ sets any encountered agent to \emph{no\_tie}.
\end{enumerate}

The output changes as follows:
\begin{itemize}
    \item Agents in \emph{no\_tie} output the winning color stored in $out$ as usual (by default behavior of \prm{}).
    \item Agents in \emph{tie} or \emph{tie-spreader} output \emph{``tie''}.
\end{itemize}

\begin{restatable}[\appref{thm:tiereport}]{theorem}{thmtiereport}\label{thm:tiereport}
    \tiereport{} solves the relative majority problem, where all agents output the color of the majority if the majority is strict.
    If there is a tie, agents do not output a color, but rather report that there is a tie. It requires $3k^3$ many states, which can be reduced to $4k^2$ if agents are not required to output the majority color, but only whether their color wins or loses.
\end{restatable}

\appendixproof{thm:tiereport}{
\thmtiereport*

The result follows from Theorem~\ref{theorem:correctness} together with the following lemma.

\begin{lemma}
  \label{lem:tie_report}
  If there is no tie in the input colors then all agents eventually are at \emph{no\_tie}.
  If there is a tie in the input colors then all agents eventually are at \emph{tie} or \emph{tie\_spreader}.
\end{lemma}

\begin{proof}
  We first assume that there is no tie.
  The only rules that sets an agent to \emph{tie\_spreader} is the first rule; this rule eventually stops triggering since all agents stop exchanging their ket past a certain time (Theorem \ref{theorem:stabilization}).
  As there is no tie, there will be a remaining agent in state $\braket{\mu}{\mu}$ past a certain time and that agent eventually sets all other agents to \emph{no\_tie}.\\
  We now assume that there is a tie.
  Past a certain time, no agent will be in state $\braket{i}{i}$ for any $i$.
  There hence exists an interaction where the last agent of the form $\braket{i}{i}$ decays into $\braket{i}{j}$ with $i \neq j$.
  When this happens that agent $\braket{i}{j}$ is at \emph{tie\_spreader} and will set all other agents at \emph{no\_tie} to \emph{tie}.
\end{proof}
}

\subsection{Sharing Ties}
\label{subsec:sharing_win}

In this subsection we would like that, even in case of a tie, all agents with a winning color output their color.
Agents with a losing color output any winning color.

The \tieshare{} protocol is defined as follows:
From the \prm{}, we keep everything apart from item (2) of the transition function (the one that updates the \emph{out} of agents). We add a variable ``status'' that can have values \emph{winner\_spreader}, \emph{winner} or \emph{loser}.
All agents start at \emph{winner}.
The variable status evolves in parallel of \prm{} and follows the following rules, where all brakets below are considered after the application of the transition rule (1) of \prm{}.

\begin{enumerate}
\item
  \label{rule:decay}
  \textbf{token generation: }An agent that exchanged kets to become $\braket{i}{j}$ with $i \neq j$ is set to \emph{winner\_spreader}.
\item 
    \textbf{token destructor: }An agent $\braket{i}{j}$ that recombines into $\braket{i}{i}$ is set to $\emph{winner}$.
\item
  \label{rule:destroy}
  \textbf{token destruction: }An agent with braket $\braket{i}{i}$ sets agents at $\braket{i}{j}$ to \emph{winner}, otherwise sets agents at $\braket{i ^{\prime}}{j}$ with $i' \notin \{i,j\}$ to \emph{loser}.
\item
  \label{rule:circle}
  \textbf{token multiplication: }If $a$ with $\braket{i}{j}$ at \emph{winner\_spreader} meets $b$ at \emph{loser} with $\braket{j}{j ^{\prime}}$, set $b$ to \emph{winner\_spreader}.

  \item
  \label{rule:up}\textbf{token ascent: }
  If an agent $a$ with $\braket{i}{j}$ at \emph{winner\_spreader} meets an agent $b$ with $\braket{i ^{\prime}}{j ^{\prime}}$ such that $[i,j]_k \subsetneq [i', j']_k$, set $b$ to \emph{winner\_spreader} and $a$ to \emph{winner} if $i=i'$ and to \emph{loser} otherwise.
  \item 
  \textbf{token simplification: } If an agent $a$ with $\braket{i}{j}$ at \emph{winner\_spreader} meets an agent $b$ at \emph{winner\_spreader} with $\braket{i ^{\prime}}{j ^{\prime}}$ such that $[i,j]_k \subseteq [i', j']_k$, set $a$ to \emph{winner} if $i=i'$ and to \emph{loser} otherwise.
\item
  \label{rule:spread}
  \textbf{information spreading: }If an agent $a$ with $\braket{i}{j}$ at \emph{winner\_spreader} meets an agent $b$ with $\braket{i}{j ^{\prime}}$, set $b$ to \emph{winner}. If $b$ is of the form $\braket{i'}{j'}$ with $i'\neq i$, and $[i', j']_k \subsetneq [i,j]_k$, set $b$ to \emph{loser}.
\end{enumerate}

In parallel, the \emph{out} value of agents is also updated:
\begin{enumerate}[label=(\roman*)]
\item If an agent $a$ with $\braket{i}$ meets an agent $b$ of the form $\braket{i'}{j'}$ with $i'\neq j'$, set \emph{out$_b$} to $i$.
    \item If an agent $a$ with $\braket{i}{j}$ at \emph{winner\_spreader} meets an agent $b$ of the form $\braket{i'}{j'}$ with $[i', j']_k \subseteq [i,j]_k$, set \emph{out$_b$} to $i$.
\end{enumerate}

We say that an agent has a token if it is at \emph{winner\_spreader}, and doesn't have a token if it it is at either \emph{winner} or \emph{loser}.

\begin{restatable}[\appref{thm:tiesharing}]{theorem}{tiesharethm}\label{thm:tiesharing}
  Let $M = \mu_0, \mu_1, \dots \mu_{r-1}$ be the colors in relative majority.
  In the \tieshare{} protocol, eventually all agents with initial color in $M$ are at \emph{winner\_spreader} or \emph{winner} while all other agents are at \emph{loser}.
  Moreover, the output of every \emph{winner} or \emph{winner\_spreader} is its initial color, while the output of \emph{loser} agents is one of the winning colors.

  The protocol uses $3k^3$ many states. This can be dropped to $3k^2$ if we only require each agent to only output whether their color win or lose.
\end{restatable}

\appendixproof{thm:tiesharing}{
\tiesharethm*
To prove this theorem, we split it into two cases, depending on whether $r =1$ or $r>1$.  In the rest of this section, we simplify the notation and write $\mu_{s}$ for $\mu_{s \bmod r}$

\begin{lemma}\label{lem:ifoneocolor}
    If $r=1$, all agents with initial color in $M$ are at \emph{winner} while all other agents are at \emph{loser}.
\end{lemma}

\begin{proof}
    Eventually the only agents with brakets $\braket{i}{i}$ are $\braket{\mu_0}{\mu_0}$, by Lemmas~\ref{lem:majowins} and~\ref{lemma:majority_color}.
  Those agents are at \emph{winner} from either the initial configuration or the \emph{token destructor} rule.
  Moreover, those agents will set all agents $\braket{\mu_0}{j}$ to winner and all agents $\braket{i}{j}$ with $i \neq \mu_0$ to loser, by the \emph{token destruction} rule.
  The \emph{out} of all agents is set to $\mu_0$ by rule (i).
\end{proof}

The idea for the case $r>1$ relies heavily on the following observation:

\begin{claim}
After Stabilization (Theorem~\ref{theorem:stabilization}), every braket $\braket{i}{j}$ held by an agent satisfies that there exists an $s$ such that $[i,j]_k \subseteq [\mu_s, \mu_{s+1}]_k$.   
\end{claim}

This claim is immediate by the definitions of the Greedy independent sets (Definition~\ref{def:greedy_independent_sets}), the corresponding Circle Braket Sets (Definition~\ref{def:braketsets}) and Lemma~\ref{lem:majowins}.

We thus show below that starting at some point in time, there exists a token in every interval $[\mu_s, \mu_{s+1}]_k$, there exists an agent $\braket{\mu_s}{\mu_{s+1}}$ that holds a token, and can spread the information for all agents in that interval.

We start by proving the following lemma. It relies on our definition of Greedy independent sets (Definition~\ref{def:greedy_independent_sets})  and the corresponding Circle Braket Sets (Definition~\ref{def:braketsets}).

\begin{lemma}\label{lem:ifmorecolors}
    Let $r>1$. Let $T(t)$ be the set of agents that have a token after interaction $t \in \mathbb N$.
    There exists an $\ell\in \mathbb N$ and a set of agents $T$ such that for all $t \ge \ell$, $T(t)= T$.
    Moreover, the multiset of brakets of all agents in $T$ is equal to $f(G_q)$. 
\end{lemma}

\begin{proof}
  Let $t$ be the interaction where the brackets of the agents stop changing (Theorem \ref{theorem:correctness}).
  The population contains agents with brakets $\braket{\mu_0}{\mu_1}$$, \braket{\mu_1}{\mu_2},$$ \dots, \braket{\mu_{r-1}}{\mu_0}$, by Lemma~\ref{lem:majowins}.
   
  We show that starting at $\ell$, we have the following invariant:
  \begin{claim}\label{claim:invariant}
      For any $t'\ge \ell$, we have that for any $s \in [1,r]$, there exists a $v\in [0,k-1]$ with $[\mu_s,v]_k \subseteq [\mu_s, \mu_{s+1}]_k$, and such that there exists an agent with braket $\braket{\mu_s}{v}$ either at \emph{winner\_spreader} or \emph{loser}.
  \end{claim}

  \begin{proof}
      Let us fix an $s$ and look at the last point in time $\tau$ where a $\bra{\mu_s}$ exchanged kets. 
      By \emph{token generation}, a token is thus created.
      We then prove the claim by induction. Assume that after interaction $\tau' \ge \tau$, we have that there exists a $v\in [0,k-1]$ with $[\mu_s,v]_k \subseteq [\mu_s, \mu_{s+1}]_k$, and such that there exists an agent with braket $\braket{\mu_s}{v}$ either at \emph{winner\_spreader} or \emph{loser}. We prove this holds as well after interaction $\tau'+1$.

      If the interaction did not involve a braket $\braket{\mu_s}{v}$ at either \emph{loser} or \emph{winner\_spreader}, the induction is immediate. We thus assume that such an agent is involved in the interaction.

     If the interaction $\tau'+1$ is either \emph{token generation, }or\emph{ token destructor}, the results holds trivially, as that same agent is set to either \emph{winner\_spreader} or \emph{loser} (It cannot be set at \emph{winner} using \emph{token destructor}, as that would imply the existence of a $\braket{\mu_s}$, which is incompatible with the fact that $r>1$ and the definition of $\tau$).

     If interaction $\tau'+1$ is a \emph{token destruction}, agent in $\braket{\mu_s}{v}$ must have been set to \emph{loser}, as the alternative (being set to \emph{winner} would mean that $\braket{\mu_s}{v}$ met $\braket{\mu_s}$, a contradiction to $\tau$ being the last point in time where $\bra{\mu_s}$ exchanged kets, as we know by Lemma~\ref{lem:majowins} that $\braket{\mu_s}$ is not in the final braket multiset). This maintains the invariant.

     If interaction $\tau'+1$ is a \emph{token ascent}, then either agent in $\braket{\mu_s}{v}$ must have been set to \emph{loser},
     or the agent it interacted with had a token and was in a state of the form $\braket{\mu_s}{ j}$ for some $j \in [0,k-1]$. 
     In that case the invariant holds after $\tau'+1$, because either $\braket{\mu_s}{v}$ is at \emph{loser}, or $\braket{\mu_s}{j}$ being at \emph{winner\_spreading}.   
     
     If interaction $\tau'+1$ is a \emph{token simplification}, then either agent in $\braket{\mu_s}{v}$ must have been set to \emph{loser},
     or the agent it interacted with had a token and was in a state of the form $\braket{\mu_s}{ j}$ for some $j \in [0,k-1]$. 
     In that case the invariant holds after $\tau'+1$, because either $\braket{\mu_s}{v}$ is at \emph{loser}, or $\braket{\mu_s}{j}$ being at \emph{winner\_spreading}.

     If interaction $\tau'+1$ is an \emph{information spreading}, either $\braket{\mu_s}{v}$ is set to \emph{loser}, which maintains the invariant, or it met some braket $\braket{\mu_s}{j}$ for some $j\in [0,k-1]$ at \emph{winner\_spreader}. In that case the invariant holds after $\tau'+1$, not necessarily due to $\braket{\mu_s}{v}$, but rather $\braket{\mu_s}{j}$ being at \emph{winner\_spreading}.      
  \end{proof}

  At time $t$, the population also contains at least one agent at \emph{winner\_spreader}, which was created by the last ket exchange.
  This token will ascend until it becomes attached to a $\braket{\mu_s}{\mu_{s+1}}$ for some $s \in [0,k-1]$.
In case there is no token associated to a $\braket{\mu_{s+1}}{j}$, by Claim~\ref{claim:invariant}, there is at least one agent in $\braket{\mu_{s+1}}{j}$ at \emph{loser}, and the $\braket{\mu_s}{\mu_{s+1}}$ agent at \emph{winner\_spreader} will generate a token at it by \emph{token multiplication}, which will then ascend as well.
We can thus show inductively that after some point, for all $s$, there exists a token in one of the $\braket{\mu_s}{\mu_{s+1}}$, using Claim~\ref{claim:invariant}, \emph{token multiplication} and \emph{token ascent}. 
(Note that \emph{token destruction} cannot happen anymore as we are after Stabilization (Theorem~\ref{theorem:stabilization}), while \emph{token simplification} ensures at least one token remains.
\emph{Information spreading} will then ensure that all agents in $\braket{\mu_s}{j}$ for some $s \in [1,r], j \in [0,k-1]$ are set to \emph{winner}, blocking any further \emph{token multiplication} to happen. \emph{Token simplification} then ensures that there is exactly one token associated to a $\braket{\mu_s}{\mu_{s+1}}$ for each $s$. 
\end{proof}

\begin{proof}[Proof of Theorem~\ref{thm:tiesharing}]
    The case for $r=1$ is handled by Lemma~\ref{lem:ifoneocolor}. 
    The case for $r\ge 2$ is mostly handled by Lemma~\ref{lem:ifmorecolors}, where we conclude by noting that agents in $T$ set all other agents to either \emph{winner} or \emph{loser} correctly using \emph{information\_spreading}, and update their \emph{out} accordingly (the key point here is that after Stabilization (Theorem~\ref{theorem:stabilization}), every braket $\braket{i}{j}$ satisfies that there exists an $s$ such that $[i,j]_k \subseteq [\mu_s, \mu_{s+1}]_k$).

    The claim on the number of states is straightforward, as each agent stores a braket that can have $k^2$ values, a status that can have three values, and eventually an output with $k$ possible values.
\end{proof}
}

\subsection{Breaking ties}
\label{subsec:tie_breaking}
Instead of reporting or sharing ties it is also possible to break them, meaning in case of a tie all agents agree on one of the colors in relative majority. 
We adapt \prm{} by adding a single bit $tie\_break\_token$ (short $tbt$) and by running additional steps in parallel to \prm{}. We call the resulting protocol \tiebreak{}.
We say that an agent $a$ carries the tie-break token if $tbt_a = 1$, while agents with $tbt = 0$ do not. 
The idea is as follows: whenever two agents exchange their kets, a tie-breaking token is generated.
When two tokens meet, one token is destroyed. 
The token is always passed along to the braket with the largest weight. 
The idea we rely on is that in case of a tie the braket with largest weight is formed of the bra and the ket of two majority colors. 
In case there is no tie, as before an agent with $\braket{i}{i}$ for the unique majority color $i \in [0,k-1]$ exists, obtaining the maximum weight of $k$.

More precisely the \tiebreak{} protocol is obtained from the \prm{} protocol by executing the following steps in sequence after applying the transition function of \prm{}. Note that agents may apply multiple steps (e.g.~\ref{prmtb:destroy} and~\ref{prmtb:pass}) in a single interaction.
\begin{enumerate}
	\item\label{prmtb:generate}\textbf{token generation: } If the agents exchange their kets (due to \prm{}) and neither one currently holds a tie-break token, then one of them generates a token (sets $tbt$ from 0 to 1).
	\item\label{prmtb:destroy}\textbf{token destruction: } If both agents hold a tie-break token, then one of them destroys its token (sets $tbt$ from 1 to 0).
	\item\label{prmtb:pass}\textbf{token transfer: } If exactly one agent holds a tie-break token, then they pass the token to the one with the larger weight. In case the weight of the agents is identical, then they pass the token to the one with the larger bra (to pass a token from $a$ to $b$, set $tbt_a$ from 0 or 1 to 0 and $tbt_b$ from 0 or 1 to 1).
\end{enumerate}
After applying the rules above, the output of agents is set:
\begin{itemize}
    \item If agent $a$ has the token, set \emph{out}$_a=$\emph{out}$_b= i_a$, where $a$ has braket $\braket{i_a}{j_a}$.
\end{itemize}

In case we do not specify which of the agents should generate or destroy a token, it does not matter for the correctness of the protocol, so w.l.o.g. we chose the initiator.

\begin{restatable}[\appref{thm:tiebreak}]{theorem}{thmtiebreak}\label{thm:tiebreak}
    \tiebreak{} solves the relative majority problem against a weakly-fair scheduler.
    If there is a tie, agents all output the same color, one of the majority ones.

    It requires $2k^3$ many states, which can be reduced to $3k^2$ if agents are not required to output the majority color, but only whether their color wins or loses.
\end{restatable}

\appendixproof{thm:tiebreak}{
\thmtiebreak*

The protocol stores three numbers between 0 and $k-1$ and the $tbt$ bit, giving $2k^3$ possible states. It is possible to remove the memory entry $out$ and add a single bit, which stores only whether and agent is in relative majority or not. For that an agent holding the tie-break token sets this bit to true for itself an all agents it meets with the same bra and to false for all other agents.
This requires $3k^2$ states (because any agent holding a tie-break token has the output bit to 1).
The overall claim follows from the following lemma.

\begin{lemma}
    \label{lem:tie_breaking}
    After some number of interactions all agents store the same color in \emph{out}, which is part of a (the) color in relative majority.
\end{lemma}

\begin{proof}
	By Theorem~\ref{theorem:stabilization} after some number of interactions $t_1$, the braket of every agent remains unchanged. In the corner case that $t_1 = 0$, meaning agents never exchange kets, only a single color is present in the population. In this case no tie-break token is ever generated and $out$ of all agents contains the unique color by initialization, producing the correct output.

	Otherwise in $t_1$ the last exchange of kets between agents occurs.
    Therefore a tie-break token is generated in $t_1$. This means at least one such token exists after $t_1$. It is, however, possible for multiple tokens to exist at this point, generated during previous interactions.
	After $t_1$ no more tie-break tokens are generated, because step~\ref{prmtb:generate} is never executed again.
	Step~\ref{prmtb:destroy} destroys one token whenever two tokens meet, decreasing the overall number of tokens. Since, however, two tokens are required for it to be executed, the number of tokens can never be decreased from 1 to 0, meaning at least one token exists in every computation step after $t_1$.
	After $t_1$ the weight of agents does not change anymore. The token transfer in combination with the destruction step then ensures that all tokens are eventually passed to the agent with maximum weight and among them to the one with the largest bra. If multiple agents with maximum weight and maximum bra exist, then the destruction step still ensures that eventually only one of them holds the tie-break token. This means after some number $t_2 \geq t_1$ of interactions, among all agents of maximum weight a unique agent $a$ exists that holds the tie-break token. Due to the output propagation, after some number $t_3 \geq t_2$ of interactions, all agents store $\bra{i_a}$ as output.

	All that is left to show is that $\bra{i_a}$ represents a color in relative majority. Associate with the population sets $G_1, G_2, \dots, G_q$ as in Definition~\ref{def:greedy_independent_sets}.
    By Lemma~\ref{lem:majowins} after $t_1$ the population consists of agents with brakets $\bigcup_{p \in [q]} f(G_p)$.
	Clearly $G_q$ contains all colors in relative majority. In case $f(G_q)$ contains a single element, then this element has the largest possible weight $k$.
    Otherwise, the maximum weight is obtained in $f(G_q)$ because $G_1 \supseteq G_2 \supseteq \dots \supseteq G_q$.
	Therefore $\bra{i_a}$ represents a color in relative majority, which concludes the proof.
\end{proof}
}

\section{Unordered setting}
\label{sec:unordered}
The \prm{} protocol, which we introduce in this work to solve relative majority, relies on numerical representations of the colors in order to compute a distance between them.
In this section we describe how to adapt \prm{} to the unordered setting, in which agents are only able to compare colors for equality and to memorize them.
We use the overall idea proposed by \citet{natale2019necessarymemorycomputeplurality} which is to combine our protocol with an ordering protocol that assigns to all agents of some initial color from  color set $C$ the same numeric label from label set $K=[k-1]$.
We, however, use another ordering protocol, that is more memory-efficient.

\subsection{The \po{} protocol}
The protocol reuses some ideas proposed by~\citet{natale2019necessarymemorycomputeplurality} and combines them with a modified version of the protocol proposed by \citet{shukai2012selfstabilizingleaderelection}, solving (a variant of) the leader election problem.
Initially all agents receive label 0.
\po{} elects one leader per color among the agents initially supporting that color (by using the asymmetry of interactions).
Whenever two leaders with different colors meet (recall that the agents can check for equality between colors), then one of them increments its label by one modulo~$k$ in case they have the same label. Non-leaders simply copy the label of their leader. Eventually there is a unique leader per color and all leaders store a different label, that they then distribute to the non-leaders.

\begin{restatable}[\appref{thrm:ordering}]{theorem}{ordering}
	\label{thrm:ordering}
	\po{} solves the ordering problem with $O(k^2)$ states.
\end{restatable}

Intuitively the proof of correctness associates a potential function with every configuration that measures how far away the configuration is from containing a valid ordering.
It decreases by a lot in case a leader chooses a previously unused label and by a little in case a leader gets closer to an unused label, until it eventually reaches 0 and an ordering is established.

\appendixproof{thrm:ordering}{
\ordering*

Formally we define the protocol as follows.\\

\textbf{Memory organization: }
The state of an agent $a$ encodes the following information:
\begin{enumerate*}
	\item $ld_a$, a single bit indicating whether $a$ is a leader
	\item $ic_a$ the initial color of $a$ (in an arbitrary representation)
	\item $d_a$ the label of $ic_a$ in binary representation using $\ceil{\log_2 k}$ bits
\end{enumerate*}

\textbf{Initialization: }
Initially every agents sets $ld = 1$, $d = 0$, and $ic$ is initialized with the agents initial color.

\textbf{Transition function: }
The transition function that updates the states of two agents $a, b \in A$ is given by the following pseudocode.
\begin{algorithmic}[1]
	\If {$ic_a = ic_b$}
		\If {$ld_a = 1$ and $ld_b = 1$}
			\State $ld_a \gets 0$
		\EndIf
		\If {$ld_a = 0$ and $ld_b = 1$ and $d_a \neq d_b$}
			\State $d_a \gets d_b$
		\EndIf
	\Else
		\If{$d_a = d_b$ and $ld_a = 1$ and $ld_b = 1$}
			\State $d_a \gets (d_a + 1)\ mod\ k$
		\EndIf
	\EndIf
\end{algorithmic}

\textbf{Output: }
The output of agent $a$ can be read directly from $d_a$.\\

We show different properties of the protocol, that we will then use to argue its correctness. For any color $c \in C$, let $supp(c)$ denote the set of agents initially supporting color $c$.
First we consider the leader election process.

\begin{lemma}
	\label{lem:pou_unique_leader}
	After some number of interactions $t_1$ each non-empty set $supp(c)$ of agents contains exactly one leader. After $t_1$ the leader bit of all agents will not change. 
\end{lemma}
\begin{proof}
	The proof is highly similar to that of Lemma 1 in \cite{natale2019necessarymemorycomputeplurality}. Some additional argumentation similar to that of proof of Lemma 7 in \cite{shukai2012selfstabilizingleaderelection} is added.
	At no point does \po{} change a leader bit from 0 to 1.
	The number of leaders can subsequently never increase.
	Let~$\#_c$ denote the number of leaders with color $c$ for some~$c \in C$.
	When two leaders of color $c$ interact, one of them is set to be a non-leader, due to line 3 of the transition function, decreasing $\#_c$ by exactly one.
	As long as there are at least two leaders of color $c$, $\#_c$ will be decremented in that way.
	Two agents are required to clear a leader bit and reduce $\#_c$.
	Therefore no interaction reduces $\#_c$ from 1 to 0 as soon as there is only a single agent left. Therefore there is exactly one leader for color $c$, whose leader bit will never be set to zero. This holds for any $c \in C$, meaning there is exactly one leader per color.
    
	No leader bit of a leader can be set to 0 after that point and no leader bit of a non-leader can be set to 1.
	Therefore no leader bit changes anymore.
\end{proof}

A configuration containing a unique leader for each non-empty set $supp(c)$ of agents, in which the leader bit of no agent will ever change again, is said to be \textit{leader stable}.
Let $L \subseteq A$ be the set of leaders in a leader stable configuration.
A configuration is said to contain a \textit{valid ordering} if it is leader stable and for any two leaders $a, b \in L$, $d_a \neq d_b$.
Let $m_d(S)$ denote the number of leaders with different colors having label $d$ in some leader stable configuration $S$. A label with $m_d(S) = 0$ in configuration $S$ is called unoccupied, unused or free.
Labels with $m_d(S) > 0$ are called used or occupied.
Let $n_d(S) = max(0, m_d(S) - 1)$.
Intuitively $n_d$ describes the number of ``excess leaders'' having label $d \in K$.
For any $d \in K$ and any configuration $S$ let:
\begin{align*}
	g_d(S) &= \min_{\substack{h \in K \\ m_h(S) = 0}} (h - d) \mod k\\
	R_d(S) &= (k^2 + g_d(S)) \cdot n_d(S)
\end{align*}
In case no $h \in K$ with $m_h(S) = 0$ exists, define $g_d(S) = 0$. For any configuration $S$ further define:
\begin{equation*}
	R(S) = \sum\limits_{d = 0}^{k-1} R_d(S)
\end{equation*}
Intuitively $R(S)$ expresses ``how far away'' configuration $S$ is from containing a valid ordering.
For that $R_d(S)$ combines the values $n_d(S)$, expressing how many excess leaders are currently using label $d$, with the value $g_d(S)$, measuring the distance of label $d$ to the next unoccupied label.
$R_d(S)$ can therefore be thought of as (an upper bound on) the number of computation steps required by any interaction sequence starting in configuration $S$ to ensure all excess leaders currently using label $d$ are assigned some unused label.
The factor $k^2$ takes into account that whenever a leader occupies a previously unoccupied label, the distance for all other labels to a free label increases.
To readers familiar with ordinals it might help to think of the first uncountable ordinal $\omega$ instead of $k^2$, expressing that an arbitrary increase in $g_d$ is compensated by a decrement by one in $n_d$.

Step by step we now show that $R$ will decrease until a configuration $S$ with $R(S) = 0$ is reached, at which point a valid ordering exists and $R$ will never increase again. First we show that $R(S) = 0$ corresponds to a valid ordering in $S$.
\begin{lemma}
	\label{lem:pou_valid_ordering_Req0}
	$R(S) = 0$ for a leader stable configuration $S\ \iff$ There is a valid ordering in configuration $S$.
\end{lemma}
\begin{proof}
	($\implies$) Assume $R(S) = 0$ for some leader stable configuration $S$. By definition of $n_d$ and $g_d$ it holds that $n_d(S) \geq 0$ and $g_d(S) \geq 0$ $\forall d \in K$. The number of colors $k \geq 1$ and thus also $k^2$ is a positive constant. Therefore $R(S) = 0$ implies $n_d(S) = 0\ \forall d \in K$. For any label $d$ there are therefore zero or one leaders having label $d$. This implies for any two leaders $a,b \in L, d_a \neq d_b$. $S$ therefore contains a valid ordering.\\
	($\impliedby$) Assume $S$ contains a valid ordering. By definition $S$ is leader stable. As it contains a valid ordering, there are no two agents $a, b \in L$, s.t. $d_a = d_b$. This means $m_d(S) < 2\ \forall d \in K$, implying $n_d(S) = 0\ \forall d \in K$ and thus $R(S) = 0$.
\end{proof}

We show next that as long as the value of $R$ is larger than zero there is some edge in the interaction graph that a scheduler might select, such that the value of $R$ decreases. This is because as long as $R$ is non-zero there are two leaders using the same label. Once they interact one of them selects a new label. Either it chooses one that no other leader is using or it at least gets closer to such an unoccupied label. Both options reduce $R$.
\begin{lemma}
	\label{lem:pou_reducing_R}
	For any leader stable configuration $S$ with $R(S) > 0$, there exists an interaction $(a,b)$, such that for configuration $S'$ with $S\xrightarrow{(a,b)}S'$ it holds that $R(S') < R(S)$.
\end{lemma}
\begin{proof}
	For $R(S) > 0$ to hold there must be some $d \in K$, such that $n_d(S) > 0$. This implies $m_d(S) \geq 2$, meaning there are at least two leaders $a$ and $b$ using label $d$.
	Consider the interaction between $a$ and $b$. In $S$ they both share label $d$. The interaction causes $a$ to update its label to $(d + 1)\ \mod\ k$ (for a shorter notation the $\mod\ k$ operation will be suppressed in the rest of the proof). Therefore $n_{d}(S') = n_{d}(S) - 1$, because in $S'$ one leader less uses $d$ than in $S$.
	There are two possible cases that can arise. Either the new label of $a$ was unoccupied, meaning $m_{d+1}(S) = 0$ or it was already occupied, meaning  $m_{d+1}(S) > 0$.
    
	For convenience define the sum of excess leaders:
	\begin{equation*}
		P(S) = \sum\limits_{d' = 0}^{k-1} n_{d'}(S)
	\end{equation*}
	
	Case $m_{d+1}(S) = 0$:
	Let $R(S) - R(S') = \Delta R$ be the difference between $R(S)$ and $R(S')$.  Express $\Delta R$ as $\Delta R = \Delta R^+ - \Delta R^-$ with $\Delta R^+ \geq 0 $ being the sum of all terms that decrease the value of $R(S')$ compared to $R(S)$, meaning they increase the difference between them, and with  $\Delta R^- \geq 0$ being the sum of all therms that increase the value of $R(S')$ compared to $R(S)$, meaning they decrease the difference.
    
	The assumption $m_{d+1}(S) = 0$ implies $m_{d+1}(S') = 1$ and therefore $n_{d+1}(S) = n_{d+1}(S') = 0$. In $S$ no agent was using label $d+1$ and in configuration $S'$ only agent $a$ is using it. Therefore in both $S$ and $S'$ there are no excess leaders in label $d$. As mentioned above $n_{d}(S') = n_{d}(S) - 1$. The value $n_{d''}$ for all other labels $d'' \in K$ remains unaffected by the interaction. Therefore $P(S') = P(S) - 1$, meaning the overall number of leaders that must select a different label decreases by one. This implies $\Delta R^+ \geq k^2$, because $k^2$ is added to the sum one time less and the difference between $R(S)$ and $R(S')$ is therefore increased by a value of at least $k^2$.
    
	Because $d+1$ is not a free label anymore in $S'$, for every label $d' \in K$ the value of $g_{d'}(S')$ might increase compared to the value of $g_{d'}(S)$. Considering the range of $g_{d'}$ the value can increase by at most $k-1$.
	With $k$ many leaders this means $\Delta R^- \leq k(k-1) = k^2 - k$, because at most $k$ times the value of $k-1$ is added to the sum\footnote{This is very crude estimate, as obviously not for all $k$ leaders the distance can be increased by $k-1$ at the same time. A tighter analysis would allow to choose a constant smaller than $k^2$ in the definition of $R(S)$ but would otherwise have no other consequences for the proof.},
	decreasing the difference in $R(S)$ and $R(S')$ by at most $k^2-k$.
    
	It now therefore holds that $\Delta R = \Delta R^+ - \Delta R^- \geq k^2 - (k^2 - k) = k > 1$. Therefore $R(S) > R(S')$.\\
	
	Case $m_{d+1}(S) > 0$: In this case it holds that $n_{d+1}(S') = n_{d+1}(S) + 1$, meaning in $S$ there already was at least one agent using label $d+1$ and in configuration $S'$ there is now one more. Together with $n_{d}(S') = n_{d}(S) - 1$ this means $P(S') = P(S)$, meaning the number of leaders that must select a new label remains unchanged.
    
	We now argue that $\forall d' \in K: g_{d'}(S) = g_{d'}(S')$, meaning the distance to the next free label does not change for any label. It holds that $m_d(S) \geq 2$ and $m_d(S') \geq 1$. Also it holds that $m_{d+1}(S) \geq 1$ and $m_{d+1}(S') \geq 2$. Neither of the labels $d$ and $d+1$ was unoccupied in $S$ and neither is in $S'$. For all other labels $d' \in K \setminus \{d, d+1\}: m_{d'}(S) = m_{d'}(S')$. Interaction $(a,b)$ thus did not lead to any previously occupied label being unoccupied or to any unoccupied label being occupied, therefore it did not change the distance to the next free label for any $d \in K$.
    
	Because $R(S) > 0$ and because there are $k$ many labels and at most $k$ many leaders, there must be some free label.
	The fact that $m_{d}(S) > 0$ and $m_{d+1}(S) > 0$ implies $g_d(S) = g_{d+1}(S) + 1$. This means that because both labels $d$ and $d+1$ are used in $S$, the next free label is one step further away for $d$ than it is for $d + 1$. Define $p = g_d(S)$. Then $g_{d+1}(S)  = p - 1$ and as argued before $g_d(S') = p$ and $g_{d+1}(S') = p -1$.
    
	Now consider the following two terms:
	\begin{align*}
          R_{d}(S') &+ R_{d+1}(S')\\
                    &= n_d(S')(k^2 + p) + n_{d+1}(S')(k^2 + p - 1)\\
											&= (n_d(S) - 1)(k^2 + p) + (n_{d+1}(S) + 1)(k^2 + p -1)\\
		R_{d}(S) + R_{d+1}(S) & = n_d(S)(k^2 + p)  + n_{d+1}(S)(k^2 + p - 1)
	\end{align*}
	Subtracting $R_{d}(S') + R_{d+1}(S')$ from $	R_{d}(S) + R_{d+1}(S)$ gives:
	\begin{multline*}
			R_{d}(S) + R_{d+1}(S) - (R_{d}(S') + R_{d+1}(S'))
			\\=\ k^2 + p - (k^2 + p -1)
			=\ 1
	\end{multline*}
	It further holds that $\forall d' \in K \setminus \{d, d+1\}: R_{d'}(S) = R_{d'}(S')$, because $d'$ stays unaffected by the interaction.
	It overall follows that $R(S') < R(S)$.
\end{proof}

As long as $R$ is non-zero, there is some interaction that decreases $R$. We now show that for all configurations there is never any interaction that increases $R$.
\begin{lemma}
	\label{lem:pou_nonincreasing_R}
	For any leader stable configuration $S$ and interaction $S \xrightarrow{(a,b)} S'$ it holds that $R(S') \leq R(S)$.
\end{lemma}
\begin{proof}
	Consider any operation of the transition function that may be executed on the interaction $(a, b)$.\\
	Case $ic_a = ic_b$: Only the label of a non-leader may be changed in this case. Non-leaders do not influence the value of $R$ in any way, implying $R(S) = R(S')$.\\
	Case $ic_a \neq ic_b$: If $d_b \neq d_a$ or $ld_a = 0$ or $ld_b = 0$, then the state of $a$ and $b$ remains unchanged, meaning $R(S) = R(S')$.
	Otherwise $d_a$ will be changed to $d_a' = (d_a + 1)\ mod\ k$.  By Lemma~\ref{lem:pou_reducing_R} this implies $R(S') < R(S)$.
\end{proof}

With all the above lemmata defined, we show correctness of \po{}.
\ordering*
\begin{proof}
	After some number of interactions $t_1$ \po{} reaches a leader stable configuration $S_{t_1}$ according to Lemma~\ref{lem:pou_unique_leader}. At that point $R(S_{t_1})$ can be defined. Obviously $R(S_{t_1})$ is a finite value.
	We show by induction over all possible values of $R(S_{t_1})$, that eventually a configuration $S_{t_2}$ will be reached with $R(S_{t_2}) = 0$.\\
	
	\textit{Induction base:} $R(S_{t_1}) = 0$. In this case $t_1 = t_2$ and the claim obviously holds.\\
	
	\textit{Induction hypothesis:} \po{} eventually reaches a configuration $S_{t_2}$ with $R(S_{t_2}) = 0$, when at $t_1$, $R(S_{t_1}) = x \geq 0$.\\
	
	\textit{Induction step:} To show: \po{} eventually reaches a configuration $S_{t_2}$ with $R(S_{t_2}) = 0$, when at $t_1$, $R(S_{t_1}) = x + 1$.\\
	According to Lemma~\ref{lem:pou_reducing_R} there exists some interaction $(a, b)$, such that $S_{t_1} \xrightarrow{(a, b)} S_{t_1+1}$ and $R(S_{t_1+1}) < R(S_{t_1})$. By Lemma~\ref{lem:pou_nonincreasing_R} no interaction any scheduler might select increases $R$. Looking at the proof of Lemma~\ref{lem:pou_nonincreasing_R} more closely it clearly also follows that any interaction that does not decrease $R$, does not change the state of any leader. Therefore no matter how often a scheduler may select any interaction that does not decrease $R$, leading to some configuration $S_{t_1}'$, the interaction $(a, b)$ still produces a configuration $S_{t_1 + 1}'$ such that $S_{t_1}' \xrightarrow{(a, b)} S_{t_1 + 1}'$ and $R(S_{t_1 + 1}') < R(S_{t_1}') = R(S_{t_1})$. A weakly fair scheduler (as well as a random or a globally fair one) will therefore eventually either select interaction $(a, b)$ or another one reducing $R$. Therefore eventually $R$ will be decreased and by induction hypothesis a configuration $S_{t_2}$ with $R(S_{t_2}) = 0$ will be reached.\\
	
	Obviously $R$ can never get negative. Using Lemma~\ref{lem:pou_nonincreasing_R} again therefore after $t_2$, $R$ will never increase and thus remain zero. In $S_{t_2}$ there is a valid ordering on the leaders, according to Lemma~\ref{lem:pou_valid_ordering_Req0}. After $t_2$ the only state changes that will occur are those of non-leaders interacting with the leaders of their color, causing the non-leader to copy the label of the leader, until finally in some configuration $S_{t_3}$, all agents of the same color agree on the same binary label, which is different from the labels of all other colors.
    
	\po{} stores the initial color of an agent, for which there are $k$ options at most, the label associated with that color, again having at most $k$ options, and a single bit indicating whether the agent is a leader. This gives $2k^2 \in O(k^2)$ possible states.
\end{proof}
}

\subsection{The \prmo{} protocol}
The output of \po{} is used as input to \prm{}. To reiterate: \po{} computes a numeric \textit{label} from set $K=[k-1]$ for an agents initial \textit{color} from (the arbitrary) set $C$.
That label is stored as the agents \textit{bra}, which functions as input to \prm{}.
One important property that is required for the correctness of \prm{} is that it maintains the invariant from Lemma~\ref{lemma:global_braket_invariant}, so for every bra $\bra{i}$ in the population for some $i \in K$, there needs to be a matching ket $\ket{i}$ in the population as well.
Care must therefore be taken when overwriting the bra of an agent with a new label.
Whenever \po{} wants to update an agents bra, we therefore make sure that it is done in a way that maintains the invariant.
We say that an agent is \textit{consistent} in case its braket is $\braket{i}{i}$ for $i \in K$ and \textit{inconsistent} otherwise.
Before applying changes to an agent's bra, we first wait for the agent to become consistent.
Then its braket can be re-initialized to a new value $\braket{j}{j}$ and Invariant~\ref{lemma:global_braket_invariant} still holds.
The idea of waiting for agents to become consistent before transmitting changes of the ordering protocol to the relative majority protocol also originates from~\citet{natale2019necessarymemorycomputeplurality}.

In order to remember what changes \po{} wants to apply to the bra of an agent until that agent becomes consistent, we introduce an \textit{ordering state}.
The ordering state stores whether \po{} wants to increment the bra, copy it from the leader or leave it unchanged.
For that it can take one of the values $I$, $CL$, $CLW$ or $U$ (for increment, copy leader, copy leader wait or unchanged).
The two separate states $CL$ and $CLW$ are required to make the following distinction: an agent in ordering state $CLW$ wants to copy the label of its leader, but first has to wait to become consistent, while an agent in ordering state $CL$ is already consistent and now only has to wait to meet its leader in order to copy its label.
Agents in ordering state $U$ will be called \textit{stable}, because they do not want to make any changes to their bra.
Agents in $I$ or $CLW$ will be called \textit{unstable} to indicate they want to change their bra, but first need to become consistent.
Finally agents in $CL$ we call \textit{pending}, highlighting that they are already consistent and only waiting to meet their leader.

Formally the \prmo{} protocol which solves relative majority in the unordered setting works as follows.\\

\textbf{Memory organization and initialization: }
Agents store all memory entries related to \prm{}: $\braket{i}{j}$, $out$ and depending on the use case bits to break, share or report ties.
Additionally there are memory entries $ic$  to store the initial color of an agent (in some representation) and the leader bit $ld$ required for \po{}. 
Agents also store an \textit{ordering state} $os \in \{U, I, CL, CLW\}$. Initially $i = j = 0$,  $os = U$ and $ld = 1$. Both $ic$ and $out$ are initialized with the input color of the agent.

\textbf{Transition function: }
When we say an agent with braket $\braket{i}{i}$ re-initializes for some value $i'$, then it sets its braket to $\braket{i'}{i'}$, $os = U$, $out = ic$ and re-initializes all variables related to reporting or breaking ties as described in the respective sections.
An interaction between agents $a$ and $b$, with brakets $\braket{i_a}{j_a}$ and $\braket{i_b}{j_b}$, results in state updates by performing the following steps in sequence:
\begin{enumerate}
	\item \label{prmo:ordering} In case $a$ and $b$ are leaders of the same color, then $a$ stops being a leader, meaning it sets $ld_a = 0$. In case $b$ is the leader of $a$, $i_a \neq i_b$ and $a$ is not pending, then $a$ sets $os_a = CLW$. In case $a$ and $b$ are leaders of different colors, have the same bra and $b$ is stable, then $a$ sets $os_a = I$. This step corresponds to \po{}.
	\item \label{prmo:copy_leader} If $a$ is pending, $a$ and $b$ have the same initial color $ic \in C$ and $b$ is a leader, then $a$ re-initializes its label with value $i_b$. Pending agents do otherwise not interact in any of the other steps.
	\item \label{prmo:become_consistent} If $a$ is unstable and inconsistent, and $j_b = i_a$, then $a$ and $b$ exchange their $ket$, making $a$ consistent.
	\item \label{prmo:become_stable} If $a$ is unstable and consistent, then in case $os_a = I$,  $a$ re-initializes for the value $(i_a + 1) \mod k$. In case $os_a = CLW$, then $a$ sets $os_a = CL$.
	\item \label{prmo:majority} $a$ and $b$ interact according to \prm{}.
\end{enumerate}

In order for agents to output elements from the set of colors $C$ instead of elements from the set of labels $K$, manipulate \prm{} as follows: Whenever some agent $a$ sets $out_a = \bra{i_b}$ for some agent $b$ (potentially $a = b$), instead set $out_a = ic_b$.\\

\begin{restatable}[\appref{theorem:unordered}]{theorem}{thmunordered}
	\label{theorem:unordered}
	\prmo{} solves relative majority in the unordered setting against a weakly-fair scheduler using $O(k^4)$ states. This can be reduced to $O(k^3)$ many states if every agent is only required to report whether their color wins or loses.
\end{restatable}

The proof of correctness essentially argues that eventually an ordering among the colors is established and that afterwards \prm{} runs on a population correctly representing the initial color distribution with the bra and ket of agents.

\appendixproof{theorem:unordered}{
\thmunordered*
We will first show properties of \prmo{}, before ultimately arguing for its correctness. On a high level the proof works as follows. Every operation the ordering step intends to make on an agent's bra is eventually carried out and corresponds precisely to ordering protocol \po{}. The steps are carried out after a finite number of interactions and in a way that preserves correctness of sub-protocol \prm{}. Once the ordering converges only \prm{} is run and the overall process stabilizes.

We subdivide this argument further into smaller parts. First we show that \prmo{} preserves \prm{}'s invariant. This is crucial to later argue for the overall convergence of \prmo{}

\begin{lemma}
	\label{lemma:prmo_invariant}
	\prmo{} maintains Invariant~\ref{lemma:global_braket_invariant}, meaning in \prmo{} for every $i \in K$, there are as many bras $\bra{i}$ in the population, as there are kets $\ket{i}$.
\end{lemma}

\begin{proof}
	Step~\ref{prmo:ordering} does not change the braket of any agent and therefore does not influence the invariant. Step~\ref{prmo:majority} updates the braket of agents according to \prm{}, which preserves the invariant. Left to consider are therefore steps~\ref{prmo:copy_leader}, \ref{prmo:become_stable} and~\ref{prmo:become_consistent}.
	
	Step~\ref{prmo:copy_leader} re-initializes pending agent $a$.
	Any pending agent is consistent, because of the following.
	Step~\ref{prmo:become_stable} is the only way for an agent to be put into ordering state $CL$. For that is must be consistent and in ordering state $CLW$. Only non-leaders are put into ordering state $CLW$ in step~\ref{prmo:ordering}. An agent in ordering state $CL$, being a non-leader, only interacts using step~\ref{prmo:copy_leader}, which re-initializes it. Therefore as long it is pending, it remains consistent. 
	The braket of $a$ therefore has the form $\braket{i_a}{i_a}$ for some $i_a$. After re-initializing $a$ stores $\braket{i_b}{i_b}$. Step~\ref{prmo:copy_leader} therefore removes a matching braket pair and adds a matching braket pair to the population. This clearly preserves Invariant~\ref{lemma:global_braket_invariant}.
	
	Step~\ref{prmo:become_stable} re-initializes consistent agent $a$ with braket $\braket{i_a}{i_a}$ to $\braket{(i_a + 1) \mod k}{(i_a + 1) \mod k}$. As before one matching braket pair is removed and one added, preserving the invariant.
	
	Step~\ref{prmo:become_consistent} exchanges the kets of two agents, obviously preserving the Invariant. The overall claim therefore holds.
\end{proof}

Next we consider the leader election process within the ordering routine to simplify subsequent proofs.

\begin{lemma}
	\label{lemma:leader_election}
	After some number of interactions there is a unique leader in every set of agents initially supporting the same color. Afterwards the leader bit of no agent changes.
\end{lemma}
\begin{proof}
	See proof of Lemma~\ref{lem:pou_unique_leader}
\end{proof}

We again say a configuration is \textit{leader stable} if it contains a unique leader per color in the population and the leader bit of no agent can change anymore. Once such a configuration is reached, further important observations can be made about the population.

\begin{lemma}
	\label{lemma:prmo_os}
	After a leader stable configuration is reached, the ordering state of any agent changes only finitely often.
\end{lemma}
\begin{proof}
	Again we call some label $i \in K$ \textit{used} in case there exists a leader that stores the label as its bra and \textit{free} otherwise.
	A leader can never be in ordering state $CL$ or $CLW$.
	To put a leader $a$ into ordering state $os_a = I$, it must interact with some leader $b$ with the same bra and $os_b = U$. In this interaction ordering state and bra of $b$ remain unchanged. Only leaders with ordering state $I$ ever change their bra. Therefore once a label is used by some leader, it will never be free, because two leaders using the label are required to put one of them into ordering state $I$, while ordering state and bra of the other one remain unaffected by the interaction.
	
	Now assume there was some leader $a$ that was changed from ordering state $U$ to ordering state $I$ $k$ many times. Then $a$ must have increased its label $k-1$ times, as a leader can go from ordering state~$I$ to ordering state~$U$ only using step~\ref{prmo:become_stable}. This means $a$ used every label once at some point. As described above no label that was once used by some leader will ever become free again. Therefore all the labels are used. However, $a$ goes into ordering state $I$ a $k^{\text{th}}$ time, meaning there are at least two leaders in the last label $a$ chose. Overall this means for every label there is at least one leader using it and there is a label that at least two leaders use. This means there are more than $k$ leaders, which is a contradiction, because, as shown before, in a leader stable configuration there are as many leaders as there are colors and there are at most $k$ colors. This contradiction implies that leaders only go into ordering state $I$ at most $k-1$ times and therefore only change their bra at most $k-1$ times. Since leaders can only alternate between ordering states $U$ and $I$, they therefore change their ordering state only finitely often.
	
	Some non-leader $a$ goes from ordering state $U$ into $CLW$ only in case the leader $b$ of its color changed its label since $a$ last interacted with it. As just shown leaders change their label at most $k-1$ many times and therefore non-leaders go into state $CLW$ at most $k-1$ times and therefore also only update their label at most $k-1$ times. Non-leaders can go from ordering state $U$ to $CLW$ to $CL$ to $U$ again, but not backwards. They therefore only update their states finitely often.
\end{proof}

The argument above is crucial to now argue that agents do not wait indefinitely to become consistent. This is important because only consistent agents can be re-initialized with a new bra.

\begin{lemma}
	\label{lemma:prmo_consistent}
	Any inconsistent and unstable agent $a$ eventually becomes consistent. 
\end{lemma}
\begin{proof}
	Let $a$ be an inconsistent and unstable agent. This means the braket of $a$ contains $\braket{i_a}{j_a}$ with $i_a \neq j_a$.
	For contradiction assume~$a$ remains inconsistent throughout the computation.
	By Invariant~\ref{lemma:global_braket_invariant}, there exists some agent $b$ with braket $\braket{i_b}{j_b}$, such that $i_b\neq j_b = i_a$.
	For every configuration let $B$ be the set of non-pending agents with $i_a$ as ket. We will now argue that under a weakly fair scheduler $a$ eventually interacts with an agent from $B$.
	
	By weak fairness $a$ interacts with all other agents infinitely often. Also, as long as $a$ is inconsistent, $B$ always contains some other agent.
	According to Lemma~\ref{lemma:prmo_os} no agent changes its ordering state infinitely often.
	Clearly this also holds in case we assume~$a$ (and potentially other agents) to remain inconsistent throughout the computation, because in case~$a$ is a non-leader it cannot make other agents unstable and in case~$a$ is a leader, it cannot make leaders of different color unstable (since for that it would need to be stable) and can make non-leaders of its color unstable only once (since it does not change its~$bra$).
	This means after some number of interactions either $a$ interacted with some agent in $B$ or $a$ remains inconsistent, but none of the other agents change their ordering state anymore.
	All agents then only interact using step~\ref{prmo:majority} of \prmo{}, as all other steps update the ordering state of at least one agent.
	By correctness of \prm{} therefore eventually the braket of all agents stabilizes, meaning no agents exchange their kets anymore (even in case pending vertices exist as they are consistent and can safely be ignored).
	At this point weak fairness guarantees that $a$ interacts with some agent in $B$. Step~\ref{prmo:become_consistent} will be executed, making $a$ consistent, which produces a contradiction. No other step of \prmo{} can change $os_a$ before $a$ is consistent.
\end{proof}

With all the above arguments in place we can now argue that the sub-protocol \po{} actually converges within \prmo{}.

\begin{lemma}
	\label{lemma:prmo_ordering}
	After some number of interactions for all pairs of agents $a, b \in A$ it holds that $\bra{i_a} = \bra{i_b} \iff ic_a = ic_b$. All agents are in ordering state $os = U$. Both bra and $os$ of no agent subsequently changes.
\end{lemma}
\begin{proof}
	Step~\ref{prmo:ordering} puts a leader into ordering state~$I$ exactly in case \po{} would increment its label, namely in case there is another (stable) leader with of different color using the same bra. Whenever an agent is put into ordering state $I$, then it either is consistent or by Lemma~\ref{lemma:prmo_consistent} eventually will be. At that point step~\ref{prmo:become_stable} carries out the increment modulo $k$ and puts the agent back into ordering state $U$.
	Step~\ref{prmo:ordering} puts a non-leader into ordering state $CLW$ exactly in case \po{} would copy the label of its leader, namely in case its leader changed its bra since they last interacted (and the agent is not already trying to carry out a copy operation). Again either the agent already is consistent or it eventually will be, at which point step~\ref{prmo:become_stable} and then step~\ref{prmo:copy_leader} ensure it copies its leaders label and sets its ordering state to $U$.
	
	By correctness of the ordering protocol \po{}, therefore eventually for two agents $a, b \in A$ it holds that $\bra{i_a} = \bra{i_b}$ if and only if they store the same initial color. On the last increment or copy operation all agents set $os = U$. With all agents of the same initial color sharing the same bra, no agent can subsequently be put into ordering state $CLW$ and therefore also not into ordering state $CL$. With all agents of different color having different bras, no agent can be put into ordering state $I$. Therefore all agents store ordering state $U$.
\end{proof}

Having shown that an ordering on the colors is eventually established, we now arrive at the final theorem.

\thmunordered*
\begin{proof}
	According to Lemma~\ref{lemma:prmo_ordering} after some number $t_1$ of interactions, the ordering process converges, $os = U$ for all agents and $\bra{i_a} = \bra{i_b} \iff ic_a = ic_b$ . At that point only step~\ref{prmo:majority} of \prmo{} is executed, meaning only \prm{} is used.
	Because Invariant~\ref{lemma:global_braket_invariant} of \prm{} is maintained, according to Lemma~\ref{lemma:prmo_invariant}, \prm{} converges for the final distribution of bras.
	This means after some number $t_2 \geq t_1$ of interactions all agents agree on the bra in relative majority and therefore on the initial color in relative majority. As described before the output can be adapted to contain a color instead of a label.
	
	For some agent $a$, \prmo{} stores $ic_a$, $\braket{i_a}{j_a}$, $out_a$, $ld_a$, $os_a$ and some number of bits to handle reporting, sharing or breaking of ties. The values $ic_a$, $i_a$, $j_a$ and $out_a$ can take $k$ values each. All values $ld_a$, $os_a$ and the bits for breaking, sharing and reporting ties can take a constant number of different values. This overall gives $O(k^4)$ possible states. Removing $out$ and adding a single bit storing whether an agent is in relative majority or not, this can be reduced to $O(k^3)$ for versions with binary win-lose output.
\end{proof}
}

\section{The Ranking protocol}
\label{sec:ranking}
\sloppy
In this section we discuss how to adapt the \textsc{Circles} protocol to solve the ranking problem under a globally faire scheduler.

For the rest of this section, we denote by $x_i$, for each $i\in [0,k-1]$, the number of agents with color $i$ in the input.

\begin{definition}
    In the \emph{ranking problem}, each agent is given as an input a color in $[0,k-1]$. We assume that no two colors have the same number of agents initially, and that all colors are represented at least once. 
    At the end of the protocol, each agent with input $i$ should output a value $j\in [0,k-1]$, such that $\card{\{\ell \in [0, k-1]: x_\ell > x_i\}}=j$, which corresponds to the ranking of the color in the input, that is, the majority color outputs $0$, the second most popular color outputs $1$, and so on.
\end{definition}

\thmranking*

\begin{remark}
    This algorithm can be easily modified to work for the case where the colors are unordered, with the same number of states.
\end{remark}

Let us first discuss the intuition and observations that lead us to such a protocol. 
First, in the \textsc{Circles} protocol, we notice that, in the very specific case where color $0$ is the most popular one, followed by $1$, then $2$, and so on until $k-1$, then the structure of the configuration at stabilization is very straightforward: each set of the Greedy Independent Sets is of the form $[0, i]$ for $i \in [0,k-1]$. 
In turn, this means that all of the brakets are of the form $\braket{i}{i+1 \bmod k}$ or $\braket{i}{0}$ for $i \in [0,k-1]$.
In fact, we can show that this is a if and only if situation:

\begin{lemma}\label{lem:orderisrecognizable}
     In the \textsc{Circles} protocol, the input is such that $x_0>x_1>\dots>x_{k-1}$ if and only if, after stabilization, all of the brakets are of the form $\braket{i}{i+1 \bmod k}$ or $\braket{i}{0}$ for $i \in [0,k-1]$.
\end{lemma}

\begin{proof}[Proof of \Cref{lem:orderisrecognizable}]
    The forwards implication is straightforward, as discussed above. For the backwards implication, note that if the ket $\ket{j}$ is only present in brakets $\braket{j-1 \bmod k}{j}$, then there are at least as many agents with bra $\bra{j-1 \bmod k}$ as agents with bra $\bra j$. 
    Therefore $x_{j-1}>x_j $ for all $j \in [1,k-1]$, and the result holds.
\end{proof}

Thus the idea is to recycle the \textsc{Circles} protocol to be in the situation where the input is the initial order of the colors. 
For that, we now introduce a $O(k^4)$ protocol where the bra and ket consist now of 2 different values: $\bra{i, \rk{i}}$ and $\ket{\rk{i}, i}$ for every $i\in [0,k-1]$. 
$\rk{i}\in [0,k-1]$ represents the (belief of) ranking of the color $i$. 
Now, when applying the \textsc{Circles} protocol, the idea is that instead of computing the weights according to the colors value, we compute them according to the rankings' value. Formally, for every $i,j \in [0,k-1]$

\begin{equation}\label{eq:newweight}
  	w(\braket{i, \rk{i}}{\rk{j},j}) =
  	\begin{cases}
  		k & \textit{if } \rk{i} = \rk{j}\\
  		(\rk{j} - \rk{i}) \bmod k& \text{otherwise}
  	\end{cases}
  \end{equation}

It remains to discuss how the value $\rk{i}$ is assigned. 
We start by assigning $\rk{i} \leftarrow i$. 
Moreover, every agent is assigned a token which makes it a leader for its color, that is, this agent is in charge of assigning the ranking of the color to the corresponding color, and broadcast any changes that it decides. 
When two tokens for the same color meet, one of them has to be deleted. When two tokens for different colors $i< j$ but same ranking $\rk{i}=\rk{j}$ meet, one of them increases its ranking by one $\rk{j}\leftarrow \rk{j}+1$. 
Eventually this ensures that there is a unique leader for each color, and all the rankings are different.

Finally, rules should be implemented to ensure that whenever the ranking assigned is not correct, agents are able to change their rankings. 
We note that if the rankings are not correct, but are unique, and the \textsc{Circles} protocol has stabilized, then there exists two agents of the form $\braket{i, \rk{i}}{\rk{j}, j}$ and $\braket{i, \rk{i}}{\rk{l}, \ell}$ with $\rk{j}\neq \rk{l}$ and $\rk{j}, \rk{l} \neq 0$, or there must exist an agent $\braket{i, \rk{i}}{\rk{i}, i}$ with $\rk{i} \neq 0$. 
We now give the lemma that hints at how agents can determine that the current ranking is wrong.

\begin{lemma}
    In the \textsc{Circles} protocol, after stabilization, if the initial colors do not satisfy $x_0>x_1>\dots>x_{k-1}$, and $x_i \neq 0 \quad \forall i$, then either:
    \begin{enumerate}
        \item There exist an agent $\braket{i}$ with $i \neq 0$, or
        \item there exist two agents of the form $\braket{i}{j}$ and $\braket{i}{\ell}$ with $j \neq \ell$ and $j, \ell \neq 0$.
    \end{enumerate}
\end{lemma}

\begin{proof}
    Assume all agents of the form $\braket{i}$ satisfy $i=0$, that is, $0$ is the majority color. 
    And assume that there exists a $j$ such that $x_j < x_{j+1}$, and let $\ell$ be the smallest such $j$.
    In particular, $\ell \neq 0$ and $x_{\ell -1}>x_\ell$.
    Therefore, all of the kets $\ket \ell$ are in agents of the form $\braket{\ell-1}{\ell}$, and at least one of such agents exists since $x_\ell \neq 0$.
    Moreover, the first greedy independent set that does not contain $\ell$ contains both $\ell -1$ and $\ell+1$.
    Therefore there exists an agent of the form  $\braket{\ell-1}{\ell+1}$.
\end{proof}

In our case, this means that whenever two agents of the form $\braket{i, \rk{i}}{\rk{j}, j}$ and $\braket{i, \rk{i}}{\rk{l}, \ell}$ with $\rk{j}\neq \rk{l}$ and $\rk{j}, \rk{l} \neq 0$ meet, they need to exchange $\rk{j}$ and $\rk{l}$. 
This can only be allowed if they respectively hold the tokens for $j$ and $\ell$.
Moreover, $\braket{i, \rk{i}}{\rk{i}, i}$ must be able to claim the ranking $\rk{i}=0$.
We therefore enforce the token of color $j$ to be attached to a ket $\ket{\rk{j}, j}$.
The token should seek the agent $\braket{i, \rk{i}}{\rk{j}, j}$ with minimal $\rk{i}$ to stick to, that is, whenever $\braket{i, \rk{i}}{\rk{j}, j}^*$  meets $\braket{\ell, \rk{l}}{\rk{j}, j}$ (where the asterisk denotes the token on the color of the ket), the token should be transferred to $\braket{\ell, \rk{l}}{\rk{j}, j}$ only if $\rk{l} < \rk{i}$.
We can then enforce that whenever $\braket{i, \rk{i}}{\rk{i}, i}^*$ meets $\braket{j, \rk{j}}{0, \ell}^*$, they exchange their rankings, that is $ \rk{l} \leftarrow \rk{i}, \rk{i}\leftarrow 0$; and when $\braket{i, \rk{i}}{\rk{j}, j}^*$ and $\braket{i, \rk{i}}{\rk{l}, \ell}^*$ meet, they also exchange their rankings, that is $\rk{j}, \rk{l}\leftarrow \rk{l}, \rk{j}$.

With all of these rules, we are done describing the protocol, and it remains to show that it is correct against a globally fair scheduler.
For that, we show below that there is only one configuration that is stable, namely the one where all the assigned rankings of the colors is correct, and where the \textsc{Circles} protocol is stabilized.
We then show that this configuration is reachable from all other valid configurations (that is, configurations reachable from the initial configuration), and this is enough to conclude.

Before that, let us describe more formally the protocol.
In this description, whenever two agents meet, if multiple cases apply, select only one case: the first case that does not leave the agents unchanged.
The $\circ$ symbol denotes the fact that the existence or not of a token is not important:

\begin{enumerate}
    \item Initialization:
    \begin{itemize}
        \item Each agent in color $i$ starts in the state $\braket{i, \rk{i}=i}{\rk{i}=i, i}^*$
    \end{itemize}
    \item Token deletion: \label{step:deletion}
    \begin{itemize}
        \item If $\braket{j, \rk{j}}{\rk{i}, i}^*$ meets $\braket{\ell, \rk{l}}{\rk{i}', i}^*$, with $\rk{j}\le \rk{l}$, then $\braket{\ell, \rk{l}}{\rk{i}', i}^*$ becomes $\braket{\ell, \rk{l}}{\rk{i}, i}$
    \end{itemize}
    \item Ranking uniqueness: \label{step:uniqueness}
    \begin{itemize}
        \item If $\braket{j, \rk{j}}{\rk{i}, i}^*$ meets $\braket{\ell, \rk{l}}{\rk{u}, u}^*$, with $\rk{i}= \rk{u}$ and $i <u$, then $\braket{\ell, \rk{l}}{\rk{u}, u}^*$ becomes $\braket{\ell, \rk{l}}{\rk{u}+1, u}^*$ (in other words, $\rk{u} \leftarrow \rk{u}+1$).
    \end{itemize}
    \item Ranking broadcast: \label{step:broadcast}
    \begin{itemize}
        \item If $\braket{i, \rk{i}}{\rk{j}, j}^*$ meets $\braket{u, \rk{u}}{\rk{j}', j}$, then $\braket{u, \rk{u}}{\rk{j'}, j}$ becomes $\braket{u, \rk{u}}{\rk{j}, j}$.
        \item If $\braket{i, \rk{i}}{\rk{j}, j}^*$ meets $\braket{j, \rk{j}'}{\rk{u}, u}^\circ$, then $\braket{j, \rk{j'}}{\rk{u}, u}^\circ$ becomes $\braket{j, \rk{j}}{\rk{u}, u}^\circ$
    \end{itemize}
    \item \textsc{Circles:} \label{circles}
    \begin{itemize}
        \item Recall the weight function from \Cref{eq:newweight}.
        \item If agent $\braket{i,\rk i}{\rk j,j}^\circ$ meets $\braket{u,\rk u}{\rk v,v}^\circ$ with $\min\{w(\braket{i,\rk i}{\rk v,v}), w(\braket{u,\rk u}{\rk j,j})\} < \min\{w(\braket{i,\rk i}{\rk j,j}), w(\braket{u,\rk u}{\rk v,v})\}$, then $\braket{i,\rk i}{\rk j,j}^\circ$ becomes $\braket{i,\rk i}{\rk v,v}^\circ$ and $\braket{u,\rk u}{\rk v,v}^\circ$ becomes $\braket{u,\rk u}{\rk j,j}^\circ$ (If there are any tokens, they are exchanged to stay with their respective kets).
    \end{itemize}
    \item Token transfer: \label{transfer}
    \begin{itemize}
     \item If $\braket{i, \rk{i}}{\rk{i}, i}$ meets $\braket{\ell, \rk{l}}{\rk{i}, i}^*$, then $\braket{\ell, \rk{l}}{\rk{i}, i}^*$ becomes $\braket{\ell, \rk{l}}{\rk{i}, i}$ and $\braket{i, \rk{i}}{\rk{i}, i}$ becomes $\braket{i, \rk{i}}{\rk{i}, i}^*$
        \item If $\braket{j, \rk{j}}{\rk{i}, i}$ meets $\braket{\ell, \rk{l}}{\rk{i}, i}^*$, with $j \neq i, \ell \neq i$ and $\rk{j}< \rk{l}$, then $\braket{\ell, \rk{l}}{\rk{i}, i}^*$ becomes $\braket{\ell, \rk{l}}{\rk{i}, i}$ and $\braket{j, \rk{j}}{\rk{i}, i}$ becomes $\braket{j, \rk{j}}{\rk{i}, i}^*$
    \end{itemize}
    \item Ranking swap: \label{swap}
    \begin{itemize}
        \item If $\braket{i, \rk{i}}{\rk{i},i}^*$ meets $\braket{u, \rk{u}}{0, v}^*$ with $v \neq i$, $\rk{i} \neq 0$, then $\braket{i, \rk{i}}{\rk{i},i}^*$ becomes $\braket{i, 0}{0,i}^*$ and $\braket{u, \rk{u}}{0, v}^*$ becomes $\braket{u, \rk{u}}{\rk{i}, v}^*$
        \item If $\braket{i, \rk{i}}{\rk{j},j}^*$ meets $\braket{i, \rk{i}}{\rk{l}, \ell}^*$ with $j, \ell \neq i$, then $\braket{i, \rk{i}}{\rk{j},j}^*$ becomes $\braket{i, \rk{i}}{\rk{l},j}^*$ and $\braket{i, \rk{i}}{\rk{l}, \ell}^*$ becomes $\braket{i, \rk{i}}{\rk{j}, \ell}^*$.
    \end{itemize}
    \item Output:
    \begin{itemize}
        \item An agent in state $\braket{i, \rk{i}}{\rk{j}, j}$ outputs $\rk{i}$.
    \end{itemize}
\end{enumerate}

We now define the configuration that we show below is stable:

\begin{definition}
    Let $X$ be the following configuration: \begin{enumerate}
        \item There exists exactly $k$ tokens, one for each $i \in [0,k-1]$. Each token is on a ket of different color.
        \item The tokens are associated to $k$ different rankings, defined by: $i \in [0,k-1]$, $\mathbf r(i):=\card{\{\ell \in [0, k-1]: x_\ell > x_i\}}$.
        \item  Every bra and ket are of the form $\bra{i, \mathbf r(i)}$ for $i \in [0,k-1]$ or $\ket{\mathbf r(j), j}$ for $j \in [0,k-1]$
        \item The \textsc{Circles} protocol is stabilized, that is, all kets are of the form either $\braket{i, \mathbf r (i)}{\mathbf r(j) = \mathbf r(i)+1, j}^\circ$ or $\braket{i, \mathbf r (i)}{\mathbf r(j) = 0, j}^\circ$
        \item All tokens are on brakets of the form $\braket{i, \mathbf r (i)=0}{\mathbf r(i) = 0, i}^*$ or $\braket{i, \mathbf r (i)}{\mathbf r(j) = \mathbf r(i) +1, j}^*$.
    \end{enumerate}
     \end{definition}

\begin{lemma}[\appref{lem:Xstable}]\label{lem:Xstable}
    $X$ is a stable configuration.
\end{lemma}

\appendixproof{lem:Xstable}{\begin{proof}
    Rule~\ref{step:deletion} cannot modify any agents as there are no two tokens on kets with the same color.
    Rule~\ref{step:uniqueness} cannot modify any agents as there are no two tokens associated to the same ranking.
    Rule~\ref{step:broadcast} cannot modify any agents as for every $i$, all bras and kets have the same ranking $\mathbf r(i)$.
    Rule~\ref{circles} cannot modify any agents by \Cref{lem:orderisrecognizable}.
    Rule~\ref{transfer} cannot modify any agents as the token for $i$ with $\mathbf r(i) = 0$ is on a braket $\braket{i, \mathbf r (i)=0}{\mathbf r(i) = 0, i}^*$ which minimizes the ranking in the bra, while for all $j \neq i$, we have that the kets containing $j$ are all of the same form: $\braket{\ell, \mathbf r (\ell)}{\mathbf r(j) = \mathbf r(\ell) +1, j}$, since a token is not transferred between agents of same braket.
    Finally, Rule~\ref{swap} cannot modify any agents, as there are two tokens on brakets with bra of ranking $0$ one of which is of the form $\braket{i, \mathbf r (i)=0}{\mathbf r(i) = 0, i}^*$, while all others are on brakets of bras of different weights. The two tokens on the same bra cannot interact because one of them has the same bra and ket, while any other pair of tokens cannot interact because they have different bras.
\end{proof}}

We finally show that $X$ is reachable from every configuration.

\begin{proposition}\label{prop:Xreachable}
    $X$ is reachable from any valid configuration.
\end{proposition}

The rest of this section is dedicated to proving this proposition. We do so by taking any valid configuration, and giving a series of intermediate configurations, each reachable from the previous configuration, until we reach $X$. Let $A$ be any valid configuration.

\begin{lemma}
    $A$ satisfies the following conditions:
    \begin{enumerate}
        \item for every $i$, there exists exactly $x_i$ many bras and $x_i$ many kets of the form $\bra{i, \cdot}$ and $\ket{\cdot, i}$ respectively.
        \item for every $i$, there exist at least one ket of the form $\ket {\cdot, i}^*$
    \end{enumerate}
\end{lemma}

\begin{proof}
    Both statement are straightforward, as the first and second elements of the bras and ket respectively are never changed, only exchanged, and as the token is passed around in a leader election fashion over the kets of the form $\ket{\cdot, i}$.
\end{proof}

\begin{lemma}
    There exists a configuration $B$ reachable from $A$ with the following conditions: 
    \begin{enumerate}
        \item for every $i$, there exists exactly $x_i$ many bras and $x_i$ many kets of the form $\bra{i, \cdot}$ and $\ket{\cdot, i}$ respectively.
        \item for every $i$, there exist \textbf{exactly} one ket of the form $\ket {\cdot, i}^*$.
    \end{enumerate}
\end{lemma}

\begin{proof}
    Again, this is straightforward using Rules~\ref{step:deletion} and~\ref{step:uniqueness}, and the fact that the second elements of kets are only exchanged.
\end{proof}

\begin{lemma}
    There exists a configuration $C$ reachable from $B$ with the following conditions: 
    \begin{enumerate}
        \item for every $i$, there exists exactly $x_i$ many bras and $x_i$ many kets of the form $\bra{i, \cdot}$ and $\ket{\cdot, i}$ respectively.
        \item for every $i$, there exist exactly one ket of the form $\ket {\cdot, i}^*$. Moreover, the pairs formed by these kets form a bijection of $[0,k-1]$ to $[0,k-1]$, that is to say, the rankings of the tokens are unique.
    \end{enumerate}
\end{lemma}

\begin{proof}
    From configuration $B$, we can select the pairs of agents to interact that only use Rule~\ref{step:uniqueness} to get the desired configuration.
\end{proof}

\begin{lemma}
    There exists a configuration $D$ reachable from $C$ with the following conditions: 
    \begin{enumerate}
        \item for every $i$, there exists exactly $x_i$ many bras and $x_i$ many kets of the form $\bra{i, \cdot}$ and $\ket{\cdot, i}$ respectively.
        \item for every $i$, there exist exactly one ket of the form $\ket {\cdot, i}^*$. Moreover, the pairs formed by these kets form a bijection of $[0,k-1]$ to $[0,k-1]$, that is to say, the rankings of the tokens are unique.
        \item The rankings of all bras and kets are consistent with the rankings assigned by the tokens.
    \end{enumerate}
\end{lemma}

\begin{proof}
    From configuration $C$, rules~\ref{step:deletion} and~\ref{step:uniqueness} do not change any agent anymore. Therefore, from $C$, it suffices for each agent with a token to meet the agents with the corresponding bra or ket to assign them the correct ranking.
\end{proof}

\begin{lemma}
    There exists a configuration $E_{-1}$ reachable from $D$ with the following conditions: 
    \begin{enumerate}
        \item for every $i$, there exists exactly $x_i$ many bras and $x_i$ many kets of the form $\bra{i, \cdot}$ and $\ket{\cdot, i}$ respectively.
        \item for every $i$, there exist exactly one ket of the form $\ket {\cdot, i}^*$. Moreover, the pairs formed by these kets form a bijection of $[0,k-1]$ to $[0,k-1]$, that is to say, the rankings of the tokens are unique.
        \item The rankings of all bras and kets are consistent with the rankings assigned by the tokens.
        \item The \textsc{Circles} protocol is stabilized on this particular ranking, that is, Rule~\ref{circles} cannot modify anymore any pair of agents.
    \end{enumerate}
\end{lemma}

\begin{proof}
    From configuration $D$, Rules~\ref{step:deletion} to~\ref{step:broadcast} cannot modify any agent anymore.
    Since the \textsc{Circles} protocol works against a weakly fair scheduler, there exists a sequence of interactions that lead to its stabilization. 
    Following this sequence, by skipping any interaction that wouldn't modify any agents, we reach $E_{-1}$.
\end{proof}

\begin{lemma}[\appref{lem:finallemma}]\label{lem:finallemma}
    For every $\ell \in [0, k-1]$, there exists a configuration $E_{\ell}$ reachable from $E_{\ell-1}$ with the following conditions: 
    \begin{enumerate}
        \item for every $i$, there exists exactly $x_i$ many bras and $x_i$ many kets of the form $\bra{i, \cdot}$ and $\ket{\cdot, i}$ respectively.
        \item for every $i$, there exist exactly one ket of the form $\ket {\cdot, i}^*$. Moreover, the pairs formed by these kets form a bijection of $[0,k-1]$ to $[0,k-1]$, that is to say, the rankings of the tokens are unique.
        \item The rankings of all bras and kets are consistent with the rankings assigned by the tokens.
        \item The \textsc{Circles} protocol is stabilized on this particular ranking, that is, Rule~\ref{circles} cannot modify anymore any pair of agents.
        \item $E_\ell$ contains the kets $\ket{\mathbf r(j), j}^*$ for all $i \le \ell$, that is, all rankings up to $\ell$ are correct.
    \end{enumerate}
\end{lemma}

\appendixproof{lem:finallemma}{\begin{proof}
    In the particular configuration $E_{\ell-1}$, none of the rules~\ref{step:deletion} to~\ref{circles} apply. It thus remains to show how Rules~\ref{transfer} and~\ref{swap} are useful.

    We know that in configuration $E_{\ell-1}$, the $\ell-1$ most popular colors have the correct assigned ranking, that is, $\ket{\mathbf r(j), j}^*$ exists in the configuration for all $j \le \ell-1$, and the \textsc{Circles} protocol has stabilized.
    Let $u$ be the $\ell$-th most popular color, and $v$ be the one assigned rank $\ell$ in $E_{\ell-1}$. 
    If $u=v$, then we are already in configuration $E_\ell$ and there is nothing to prove. We thus assume that $u \neq v$.

    Let $i_1, \dots, i_{\ell-1}$ be the $\ell-1$ most popular colors.
    Since the \textsc{Circles} protocol is stabilized, we know that there is a greedy independent set with 
    the elements $i_1, \dots, i_{\ell-1}, u$.
    \begin{itemize}
        \item If $\ell=0$, then this is simply one element, and there is a braket $\braket{u, \rk{u}}{\rk{u}, u}^\circ$. Through Rule~\ref{transfer}, this braket can get the token for $u$, and through Rule~\ref{swap}, can obtain a newly assigned ranking of $0$.
        \item Otherwise, if $\ell \neq 0$, since $x_{i_{\ell-1}}>x_u$, we have that every greedy independent set containing $u$ also contains $i_{\ell-1}$. 
    Therefore, coupled with the fact that there exists a greedy independent set with $i_1, \dots, i_{\ell-1}, u$, the bra with smallest ranking associated to a ket of the form $\ket{\cdot, u}$ is $\bra{i_{\ell-1}, \ell-1}$. Through Rule~\ref{transfer}, this braket $\braket{i_{\ell-1}, \ell-1}{\cdot, u}$ can get the token for $u$.
    Similarly, we have that every greedy independent set containing $v$ also contains $i_{\ell-1}$. 
    Therefore, coupled with the fact that there exists a greedy independent set with $i_1, \dots, i_{\ell-1}, v, \dots$, the bra with smallest ranking associated to a ket of the form $\ket{\ell, v}$ is $\bra{i_{\ell-1}, \ell-1}$. Through Rule~\ref{transfer}, this braket $\braket{i_{\ell-1}, \ell-1}{\ell, v}$ can get the token for $v$.
    And thus through Rule~\ref{swap}, $\braket{i_{\ell-1}, \ell-1}{\cdot, u}^*$ and $\braket{i_{\ell-1}, \ell-1}{\ell, v}^*$ meet and exchange rankings, with $u$ getting the ranking $\ell = \mathbf r(u)$.
    \end{itemize}
    
    The tokens who have had their rankings changed can now broadcast the new value of their rankings through Rule~\ref{step:broadcast}, and then any steps needed to stabilize the \textsc{Circles} protocol can happen through Rule~\ref{circles}.    
\end{proof}}

As $E_{k-1}=X$, this concludes the proof of \Cref{prop:Xreachable}. 

\begin{proof}[Proof of \Cref{thm:ranking}]
    Since $X$ is a stable configuration, by \Cref{lem:Xstable}, and is reachable from any other configuration by \Cref{prop:Xreachable}, we know that $X$ is the only stable configuration, and is always reached under a globally fair scheduler.
\end{proof}

\section{Future work}
\label{sec:future_work}
Our protocol solves relative majority always-correctly with $O(k^3)$ states, narrowing the gap with the $\Omega(k^2)$ lower bound.
Its versatility further allows computing the full ranking, making it a promising building block for more complex distributed tasks.
During our work, both \prm{} and an independent protocol based on a different logic hit $\Omega(k^3)$ states as a hard barrier.
We therefore believe the most fruitful direction is strengthening the $\Omega(k^2)$ lower bound.
Such improvements are rare in the population protocol model, and we hope this work motivates the development of new lower bound techniques with broader impact on population protocols.

\newcommand\blfootnote[1]{%
  \begingroup
  \renewcommand\thefootnote{}\footnote{#1}%
  \addtocounter{footnote}{-1}%
  \endgroup
}

\begin{acks}
 Funded by the European union. Views and opinions expressed are however those of the author(s) only and do not necessarily reflect those of the European Union or the European Research Council Executive Agency. Neither the European Union nor the granting authority can be held responsible for them. This project has received funding from the European Research Council (ERC) under the European Union's Horizon 2020 research and innovation programme (MoDynStruct, No. 101019564)  \includegraphics[height=8pt]{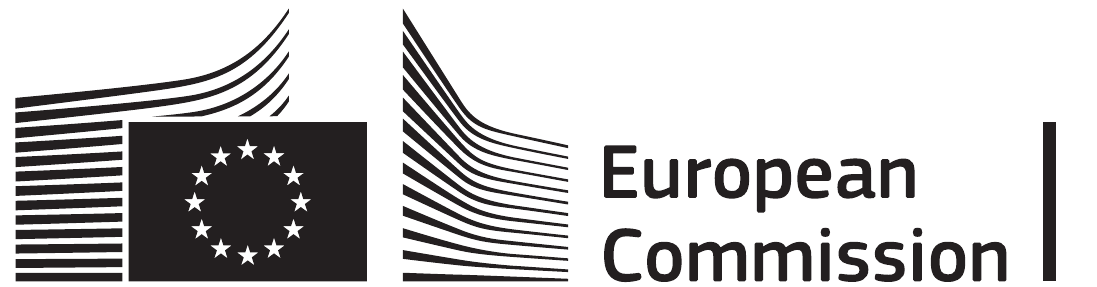} and the Austrian Science Fund (FWF) grant \href{https://www.doi.org/10.55776/I5982}{DOI 10.55776/I5982}.
For open access purposes, the author has applied a CC BY public copyright license to any author-accepted manuscript version arising from this submission.
\end{acks}
\balance
\bibliographystyle{ACM-Reference-Format}
\bibliography{references_stripped}

\section{Missing Proof Details}

\appendixProofs

\end{document}